\documentclass[submission]{eptcs}

\providecommand{\event}{EXPRESS/SOS 2018} 

\newif\iflong\longfalse
\newif\ifscreen\screenfalse

\usepackage{breakurl}

\usepackage[T1]{fontenc}
\usepackage{lmodern} 
\usepackage[utf8]{inputenc} 
\usepackage{newtxtext} 

\usepackage{amsmath}
\usepackage{amssymb}
\usepackage{mathrsfs}

\usepackage{stmaryrd}
\SetSymbolFont{stmry}{bold}{U}{stmry}{m}{n}

\usepackage{amsthm}

\newtheorem{definition}{Definition}[section]
\newtheorem{example}[definition]{Example}
\newtheorem{lemma}[definition]{Lemma}
\newtheorem{theorem}[definition]{Theorem}

\usepackage{holindex}
\usepackage{alltt}
\usepackage{proof}

\newcommand{\HOLTokenStrongEQ}{$\sim$}
\newcommand{\HOLTokenWeakEQ}{$\approx$}
\newcommand{\HOLTokenObsCongr}{$\approx^c\!$}
\newcommand{\HOLTokenEPS}{$\overset{\epsilon}{\Rightarrow}$}
\newcommand{\HOLTokenTransBegin}{$-$}
\newcommand{\HOLTokenTransEnd}{$\rightarrow$}
\newcommand{\HOLTokenWeakTransBegin}{$=$}
\newcommand{\HOLTokenWeakTransEnd}{$\Rightarrow$}

\newcommand{\HOLTokenContracts}{$\succeq_{bis}\!$}
\newcommand{\HOLTokenObsContracts}{$\succeq^c_{bis}\!$}

\usepackage{environ}
\NewEnviron{HOLTrans}{\overset{\BODY}{\longrightarrow}}
\NewEnviron{HOLWeakTrans}{\overset{\BODY}{\Longrightarrow}}

\usepackage{listings}
\usepackage{graphicx}
\usepackage[all,cmtip]{xy}
\usepackage{color}

\newcommand*\hl{}

\input{commonmacros.sty}
\input{davidemacro.sty}
\usepackage{DSarrow} 

\hypersetup{pdftitle={Unique Solutions of Contractions, CCS, and their HOL Formalisation},
	pdfauthor={Chun Tian; Davide Sangiorgi},
	pdfsubject={Concurrency Theory},
	pdfkeywords={theorem proving, process calculi, coinduction, unique solution of equations, congruence}}

\title{Unique Solutions of Contractions, CCS,\\
and their HOL Formalisation}

\author{Chun Tian
\institute{Fondazione Bruno Kessler%
\thanks{\hl{Part of this work was
  carried out when the first author was studying at Universit\`a di Bologna.}}
\\Trento, Italy}
\email{ctian@fbk.eu}
\and 
Davide Sangiorgi
\institute{Universit\`a di Bologna and INRIA
\\Bologna, Italy}
\email{davide.sangiorgi@unibo.it}
}
\def\titlerunning{Unique Solutions of Contractions, CCS, and their HOL Formalisation}
\def\authorrunning{Chun Tian \& Davide Sangiorgi}

\begin{document}
\maketitle
\begin{abstract}
  The unique solution of contractions is a proof technique for
  bisimilarity that overcomes certain syntactic constraints of
  Milner's ``unique solution of equations'' technique.  The paper
  presents an overview of a rather comprehensive formalisation of the
  core of the theory of CCS in the HOL theorem prover (HOL4), with a
  focus towards the theory of unique solutions of contractions.  (The
  formalisation consists of about 20,000 lines of proof scripts in
  Standard ML.)  Some refinements of the theory itself are obtained.
  In particular we remove the constraints \hl{on summation,
  which must be weakly-guarded,} by moving to \emph{rooted
  contraction}, that is, the coarsest
  precongruence contained in the contraction preorder.

\end{abstract}

\section{Introduction}

A prominent proof method for bisimulation, put forward by Robin Milner and widely used in his
landmark CCS book \cite{Mil89} is the
\emph{unique solution of equations}, whereby two tuples of processes are
componentwise bisimilar if they are solutions 
of the same system of equations.
This method is important in verification techniques and tools
based on algebraic reasoning \cite{BaeBOOK,theoryAndPractice,RosUnder10}. 

In the versions of Milner's unique solution theorems for proving that all
solutions are weakly (or rooted) bisimilar (in practice these are the most
relevant cases), however,
Milner's proof method has severe syntactical limitations, such that
the equations must be ``guarded and sequential,'' that is, the
variables of the equations may only be used underneath a visible
prefix and proceed, in the syntax tree, only by the sum and prefix operators.
One way of overcoming such limitations is to replace equations
 with special inequations called
\emph{contractions} \cite{sangiorgi2015equations,sangiorgi2017equations}. Contraction is a
preorder that, roughly, places some efficiency
constraints on processes.  The uniqueness of solutions of a system of contractions
is defined as with systems of equations: any two solutions must be bisimilar.
The difference with equations is in the meaning of a solution:
in the case of contractions the solution is evaluated with respect to
the contraction preorder, rather than bisimilarity. 
With contractions, most syntactic limitations of the unique-solution theorem can be
removed.  One constraint that still remains in
\cite{sangiorgi2017equations} (in which the issue is bypassed using a more
restrictive CCS syntax)
is the occurrences of direct sums, due to the failure of the
substitutivity of contraction under direct sums.

The main goal of the work described in this paper is a rather
comprehensive formalisation of the core of the theory of CCS in the HOL
theorem prover (HOL4),  with a focus on the theory of unique solutions of contractions.
The formalisation, however, is not confined to the theory of unique
solutions of equations, but embraces a significant portion the theory of CCS \cite{Mil89}
(mostly because the theory of unique solutions relies on a large number of more fundamental results).
Indeed the formalisation encompasses the basic properties of strong and weak
bisimilarity (e.g. the fixed-point and substitutivity properties), the
basic properties of
rooted bisimilarity (the congruence induced by weak
bisimilarity, also called observation congruence), and
their algebraic laws. Further extensions (beyond Nesi
\cite{Nesi:1992ve}) include four versions of ``bisimulation up to''
techniques (e.g., bisimulation up-to bisimilarity) \cite{Mil89,sangiorgi1992problem}, and the
expansion and contraction preorder (two
efficiency-like refinements of weak bisimilarity). Concerning rooted bisimilarity, the formalisation
includes Hennessy Lemma and Deng Lemma (Lemma 4.1 and 4.2 of
\cite{Gorrieri:2015jt}),
 and two long proofs saying the rooted bisimilarity is the coarsest (largest)
 congruence contained in (weak) bisimilarity: one following Milner's
 book \cite{Mil89}, with the hypothesis that no processes can use up
 all labels;
the other without such hypothesis, essentially formalising van Glabbeek's paper \cite{van2005characterisation}.
Similar theorems are proved for the rooted contraction preorder.
In this respect, the work is an extensive experiment with the use of the HOL theorem prover and its
most recent developments, including a package for expressing coinductive definitions.

From the view of CCS theory, this formalisation has offered us the possibility of
further refining the theory of unique solutions of
equations, as formally proving a previously known result gives us a
chance to see \emph{what's really needed} for establishing that result.
In particular, the existing theory \cite{sangiorgi2017equations} has
placed limitations on the body of the contractions due to the
substitutivity problems of weak bisimilarity and other behavioural relations with respect
to the sum operator.
We have thus refined the contraction-based proof technique, by moving to  
\emph{rooted contraction}, that is, the coarsest precongruence contained in the contraction
preorder. The resulting unique-solution theorem is now valid for
\emph{rooted bisimilarity} (hence also for bisimilarity itself), and places no 
constraints on the occurrences of sums.


Another advantage of the formalisation is 
that we can take advantage of results about different 
equivalences or preorders that share a similar  proof structure. 
Examples are: the results that rooted bisimilarity and rooted contraction are,
respectively, the coarsest congruence contained in weak bisimilarity 
and the coarsest precongruence contained in the contraction  preorder; 
the result about unique solution of equations for weak bisimilarity that uses the
contraction preorder as an auxiliary relation, and other unique solution results (e.g., 
the one for rooted in which
the auxiliary relation is rooted contraction); various forms of enhancements of the bisimulation
proof method (the `up-to' techniques).
In these cases,  moving between proofs there are only a few places in
which the HOL proof scripts have to be modified.
Then the successful termination of the proof gives us a guarantee that the proof is
complete and trustworthy, removing the risk 
of overlooking or missing details as in hand-written proofs.



\paragraph{Structure of the paper} 
Section~\ref{ss:ccs} presents basic background materials on CCS,
including its syntax, operational semantics, bisimilarity and rooted
bisimilarity. Section~\ref{s:eq} discussed equations and contractions.
 Section~\ref{ss:new} presents rooted contraction and the related
 unique-solution result for rooted bisimilarity. Section~\ref{s:for}
 \hl{highlights} our formalisation in HOL4. Finally,  Section~\ref{s:rel} and
 \ref{s:concl} discuss related work, conclusions,  and  a few
 directions for future work.

\section{CCS}
\label{ss:ccs}

We assume  a possibly infinite set of names $\mathscr{L} = \{a, b,
\ldots\}$ forming input and $\overline{\mbox{output}}$ actions, plus a special invisible
action $\tau \notin \mathscr{L}$, and a set of variables $A,B,
\ldots$ for defining recursive behaviours.
Given a deadlock $\nil$, the class of CCS processes is then inductively defined from $\nil$ by the operators
of prefixing, parallel composition, summation (binary choice), restriction, recursion and \hl{relabeling}:
\begin{equation*}
\begin{array}{ccl}
\mu  & := &  \tau \hspace{.3pt} \; \midd \; a  \; \midd \;  \outC a  \\
P  & := &  \nil \; \midd \;  \mu . P \; \midd \;  P_1 |  P_2 \; \midd  \;
P_1 + P_2 \; \midd 
  (\res a\!)\, P  \;  \midd \;  A \; \midd \; \recu A  P
\; \midd \; P\; [r\!\!f]  
\end{array}
\end{equation*}
The operational semantics of CCS is given by means of
a Labeled Transition System (LTS), shown in Fig.~\ref{f:LTSCCS} as SOS
rules (the symmetric version of the two rules for
parallel composition and the rule for sum are omitted).
A CCS expression uses only \emph{weakly-guarded sums} if all occurrences of
the sum operator are of the form $\mu_1.P_1 + \mu_2.P_2 + \ldots
+ \mu_n.P_n$, for some $n \geq 2$.
 The \emph{immediate derivatives} of a
process $P$ are the elements of the set $\{P' \st P \arr\mu P' \mbox{
  for some $\mu$}\}$.
\ifscreen \begin{figure*}[h] \else \begin{figure*}[t] \fi
\begin{center}
\vskip .1cm
 $\displaystyle{  \over  \mu.  P    \arr\mu
P } $  $ \hb$   
\hskip .5cm
 $\displaystyle{   P \arr\mu   P' \over   P + Q   \arr\mu
P'  } $  $ \hb$   
\hskip .5cm
 $\displaystyle{   P \arr\mu   P' \over   P | Q   \arr\mu
P' | Q } $  $ \hb$   
\hskip .3cm
  $\; \;$  $\displaystyle{ P \arr{ a}P' \hk \hk  Q
\arr{\outC a }Q'  \over     P|  Q \arr{ \tau} P'
|  Q'  }$ 
\\
\vspace{.2cm}
$\displaystyle{ P \arr{\mu}P' \over
 (\res a\!)\, P   \arr{\mu} (\res a\!)\, P'} $ $ \mu \neq a, \outC a$
$ \hb$
$\displaystyle{ P \sub {\recu A P} A \arr{\mu}P' \over
 \recu A P   \arr{ \mu} P'  } $
\hskip .5cm  
$\displaystyle{ P \arr{\mu} P' \over
 P \;[r\!\!f] \arr{r\!\!f(\mu)} P' \;[r\!\!f]} $ $\forall a.\, r\!\!f(\outC a) = \overline{r\!\!f(a)}$
$ \hb$ 
\end{center}
\caption{\hl{Structural Operational Semantics} of CCS}
\label{f:LTSCCS}
\end{figure*}
Some standard notations for transitions: $\Arr\epsilon$ is the 
reflexive and transitive closure of $\arr\tau$, and 
$\Arr \mu $ is $\Arr\epsilon \arr\mu \Arr\epsilon$ (the
composition of the three relations).
Moreover,   
$ P \arcap \mu P'$ holds if $P \arr\mu P'$ or ($\mu =\tau$ and
$P=P'$); similarly 
$ P \Arcap \mu P'$ holds if $P \Arr\mu P'$ or ($\mu =\tau$ and
$P=P'$).
We write $P \:(\arr\mu)^n P'$ if $P$ can become $P'$ after performing
$n$ $\mu$-transitions. Finally, $P \arr\mu$ holds if there is $P'$
with $P \arr\mu P'$, and similarly for other forms of transitions.

\paragraph{Further notations}
Letters  $\R$, $\S$ range over relations.
We use infix notation for relations, e.g., 
$P \RR Q$ means that $(P,Q) \in \R$.
We use a tilde to denote a tuple, with countably many elements; thus
the tuple may also be infinite.
 All
notations  are  extended to tuples componentwise;
e.g., $\til P \RR \til Q$ means that $P_i \RR Q_i$, for  each  
component $i$  of the tuples $\til P$ and $\til Q$.
And $\ct{\til P}$ is the process obtained by replacing each hole
$\holei i$ of the  context $\qct$ with $P_i$.  
We write $
\ctx \R$ for the closure of a relation under contexts. Thus $P\: \ctx \R\: Q$
means that there are context $\qct$ and tuples $\til P,\til Q$ with
$P =  \ct{\til P}, Q =  \ct{\til Q}$ and $\til P \RR \til Q$.
We use the symbol $\DSdefi$ for abbreviations. For instance, $P \DSdefi G $, where
$G$ is some expression, means that $P$ stands for the expression $G$.
If $\leq$ is a preorder, then  $\geq$  is its inverse (and
conversely).

\subsection{Bisimilarity and rooted bisimilarity}
\label{ss:BiEx}

The equivalences we consider here are mainly \emph{weak} ones, in that they
abstract from the number of internal steps being performed:
\begin{definition}
\label{d:wb}
A process relation ${\R}$ is a \textbf{bisimulation} if, whenever
 $P\RR Q$, 
we have:
\begin{enumerate}
\item $P \arr\mu P'$ implies that there is $Q'$ such that $Q \Arcap \mu Q'$ and $P' \RR Q'$;\vspace{-4pt}
\item $Q \arr\mu Q'$,implies that there is $P'$ such that $P \Arcap
  \mu P'$ and $P' \RR Q'$\enspace.
\end{enumerate}  
 $P$ and $Q$ are \textbf{bisimilar},
written as $P \wb Q$, if $P \RR Q$ for some bisimulation $\R$.
\end{definition}

We sometimes call bisimilarity the \emph{weak} one, to
distinguish it from \emph{strong} bisimilarity ($\sim$),
obtained by replacing in the above definition   the weak answer $
Q\Arcap\mu Q'$ with the strong  $Q \arr \mu Q'$.
Weak bisimilarity is not preserved by the sum operator (except for
guarded sums). For this, Milner introduced \emph{observational congruence}, also called \emph{rooted
  bisimilarity} \cite{Gorrieri:2015jt,Sangiorgi:2011ut}:
\begin{definition}
\label{d:rootedBisimilarity}
Two processes $P$ and $Q$ are \textbf{rooted bisimilar}, written as $P
\rapprox Q$, if we have:
\begin{enumerate}
 \item  $P \arr\mu P'$ implies that there is $Q'$ such that $Q
   \Arr\mu Q'$ and $P' \wb Q'$;
 \item  $Q \arr\mu Q'$ implies that there is $P'$ such that $P
   \Arr\mu P'$ and $P' \wb Q'$\enspace.
\end{enumerate}
\end{definition}

\begin{theorem}
\label{t:rapproxCongruence}
$\rapprox$ is a congruence in CCS, and it is the
coarsest congruence contained in $\approx$.
\end{theorem}



\section{Equations and contractions}
\label{s:eq}

\subsection{Systems of  equations}
\label{ss:SysEq}

Uniqueness of  solutions of equations \cite{Mil89} intuitively says that if  a context $\qct$ obeys
certain  conditions, 
then all processes $P$  that satisfy the equation $ P \wb \ct P$ are
bisimilar with each other.
We need variables to write equations. We  use
 capital
letters  $X,Y,Z$
 for  these variables and call them \emph{\behav\  variables}.
 The body of an equation is a CCS expression
possibly containing \behav\  variables. Thus such expressions, ranged
over by $E$, live in the CCS
grammar extended with \behav\  variables.

\begin{definition}
Assume that, for each $i$ of 
 a countable indexing set $I$, we have variables $X_i$, and expressions
$E_i$ possibly containing  such variables. 
Then 
$\{  X_i = E_i\}_{i\in I}$
is 
  a \emph{system of equations}. (There is one equation for each variable $X_i$.)
\end{definition}

We write $E[\til P]$ for the expression resulting from $E$ by
replacing each variable $X_i$   with the process $P_i$, assuming
$\til P$ and $\til X$ have the same length. (This is syntactic
replacement.) 
\begin{definition}
Suppose  $\{  X_i = E_i\}_{i\in I}$ is a system of equations: 
\begin{itemize}
\item
 $\til P$ is a \emph{solution of the 
system of equations  for $\wb$} 
if for each $i$ it holds
that $P_i \wb E_i [\til P]$;

\item it 
 has \emph{a unique solution for $\wb$}  if whenever 
 $\til P$ and $\til Q$ are both solutions for $\wb$, then $\til P \wb
 \til Q$. 
\end{itemize} 
 \end{definition}




For instance, the solution of the equation 
$ X = a. X$ 
is  the process
$R \DSdefi \recu A {\, (a. A)}$, and for any other solution $P$ we have $P \wb R$.
In contrast, the equation 
 $X = a|  X$ has solutions that may be quite different, \hl{for instance,
 $K$ and $K | b$, for $K \DSdefi \recu K {\, (a. K)}$. (Actually any process capable of
continuously performing $a$ actions (while behaves arbitrarily on
other actions) is a solution for $X = a|  X$.)}
 



\begin{definition}[guardedness of equations]
\label{def:guardness}
A system of equations 
$\{  X_i = E_i\}_{i\in I}$
 is 
\begin{itemize}
\item \emph{weakly guarded} if, in each $E_i$, each occurrence of
  an \behav\  variable is underneath a prefix;

\item \emph{(strongly) guarded} if, in each $E_i$, each occurrence of
  an \behav\  variable is underneath a \textbf{visible} prefix;

\item \emph{sequential} if, in each $E_i$, each of its subexpressions with occurrence of
  an \behav\  variable, apart from the variable itself, is in forms of
  prefixes or sums.
\end{itemize}
\end{definition}




\begin{theorem}[unique solution of equations, \cite{Mil89}]
\label{t:Mil89}
A system of guarded and sequential equations (without direct sums)
$\{  X_i = E_i\}_{i\in I}$ has a unique solution for $\wb$.
\end{theorem}

To see the  need of the sequentiality  condition, consider
 the equation (from \cite{Mil89}) $X = \res a (a. X | \outC a)$
where $X$ is guarded but not sequential. Any \hl{process that} does not use $a$ is a solution.

\subsection{Contractions}
\label{s:mcontr}

The constraints on the unique-solution Theorem~\ref{t:Mil89} can be 
weakened if we move from equations to a special kind of inequations called
  \emph{contractions}.

Intuitively, the bisimilarity contraction $\mcontrBIS$ is a preorder in which 
$P \mcontrBIS Q  $  holds  if $P \wb Q$ and, in addition, 
$Q$ has the \emph{possibility} of being at least as efficient as $P$ (as far as
$\tau$-actions performed). 
Process $Q$, however, may be nondeterministic and may have other ways
of doing the same work, and these could be  slow (i.e., involving
more $\tau$-steps than those performed by $P$).

\begin{definition}
\label{d:BisCon}
A process relation ${\R}$ 
 is a \textbf{(bisimulation) contraction} if, whenever
 $P\RR Q$, 

\begin{enumerate}
\item   $P \arr\mu P'$ implies there is $Q'$ such that $Q \arcap \mu
  Q'$ and $P' \RR Q'$;
\item $Q \arr\mu Q'$   implies there is $P'$ such that $P \Arcap \mu
 P'$ and $P' \wb Q'$\enspace.
\end{enumerate}
\textbf{Bisimilarity contraction}, written as $P \mcontrBIS Q$ ($P$
\textbf{contracts to} $Q$), if $P\ \R\ Q$ for some contraction $\R$.
\end{definition}

In the first clause $Q$ is required to match $P$'s challenge
transition with at most one transition.
This makes sure that $Q$ is capable of mimicking $P$'s
work at least as efficiently as $P$. 
In contrast, the second clause of Definition~\ref{d:BisCon}, on the
challenges from $Q$, entirely ignores efficiency: it is the same
clause of  weak bisimulation~--- the final derivatives are even required
to be related  by $\wb$, rather than by $\R$.

Bisimilarity  contraction is coarser than 
 the \emph{expansion relation} 
$\expa$ \cite{arun1992efficiency,sangiorgi2015equations}.
This is a
preorder widely used in proof techniques for bisimilarity and that 
intuitively refines bisimilarity by 
 formalising the idea of `efficiency' between processes.
Clause (1) is the same in the two
preorders. But in clause (2) expansion uses 
$P \Arr \mu P'$, rather than $P \Arcap \mu P'$; 
 moreover with
contraction the final derivatives are simply required to be bisimilar.
An expansion 
$P \expa Q$
tells us  that $Q$ is always at least as efficient as $P$, whereas  the
 contraction $P \mcontrBIS Q$  just says that $Q$ has the  possibility of
being at least as efficient as $P$. 


\begin{example}
\label{exa:contr}
We have 
 $ a \not  \mcontrBIS \tau. a$. However,
$a+ \tau . a \mcontrBIS a$, as well as its converse, 
$  a \mcontrBIS a +
\tau . a $. Indeed, if $P \wb Q$ then 
$  P  \mcontrBIS P +Q$. The last two relations do not hold with 
$\expa$, which explains the strictness of the inclusion
 ${\expa} \subset {\mcontrBIS}$. 
\end{example} 


\hl{Like} (weak) bisimilarity and expansion, contraction is
preserved by all operators but (direct) sum.

\subsection{Systems of contractions}
\label{ss:SysContr}

A \emph{system of contractions} is defined as a system of equations,
except that the contraction symbol $\mcontr$ is used in the place of
the equality symbol $=$. Thus a system of contractions is a set 
$\{  X_i \mcontr E_i\}_{i\in I}$
where $I$ is an  indexing set and expressions
$E_i$  may contain the  \behavC\  variables 
$\{  X_i\}_{i\in I}$.

\begin{definition}
\label{d:uniContra}
Given a system of contractions 
$\{  X_i \mcontr E_i\}_{i\in I}$, 
 we say that:
\begin{itemize}
\item $\til P$ is a \emph{solution (for $\mcontrBIS$) of the 
 system of contractions} if $\til P \mcontrBIS \til E [\til P]$;
\item the system has \emph{a unique solution (for $\approx$)}
if $\til P \approx \til Q$ whenever $\til P$ and $\til Q$ are both solutions.
\end{itemize}
\end{definition}

The guardedness of contractions follows Def.~\ref{def:guardness} (for equations).


\begin{lemma}
\label{l:uptocon}
Suppose $\til P$ and $\til Q$ are solutions  for $\mcontrBIS$
 of a system of weakly-guarded contractions that uses 
weakly-guarded sums.
For any context $\qct$  that uses 
weakly-guarded sums,
if  $\ct{\til P}\Arr{\mu}  R$,
 then 
there is a context $\qctp$  that uses 
weakly-guarded sums
such that $R \mcontrBIS \ctp{\til P}$ and $\ct{\til Q} \Arcap{\mu}
 \wb \ctp{\til Q}$.\footnote{There's no typo here: $\ct{\til Q} \Arcap{\mu} \wb \ctp{\til
     Q}$ means $\exists {\til R}.\; \ct{\til Q} \Arcap{\mu} {\til R}
   \wb \ctp{\til Q}$. Same as in Lemma~\ref{l:ruptocon}.}
\end{lemma}

\begin{proof}{(sketch from \cite{sangiorgi2017equations})}
Let $n$ be the length of the transition $\ct{\til P}\Arr\mu R$  (the
number of `strong steps' of which it is composed), and  
let $\ctpp {\til P}$ and $\ctpp {\til Q}$  be the processes obtained
from  $\ct {\til P}$ and $\ct {\til Q}$ by unfolding the definitions
of the contractions $n$ times. Thus in $\qctpp$ each hole is
underneath at least $n$ prefixes, and cannot contribute to an action
in the first $n$ transitions; moreover all the contexts have only
weakly-guarded sums.

We have $\ct{\til P} \mcontrBIS \ctpp{\til P}$, and 
$\ct{\til Q} \mcontrBIS \ctpp{\til Q}$, 
 by the substitutivity  properties of $\mcontrBIS$ (we exploit here
 the syntactic constraints on sums). Moreover,
 since each hole of the  context $\qctpp$ is underneath at least $n$
 prefixes, applying  
the definition
 of $ \mcontrBIS$ on the transition 
 $\ct{\til P}\Arr{\mu}  R$, we infer the existence
 of $\qctp$ such that 
$
\ctpp{\til P}\Arcap{\mu} \ctp{\til P} \mexpaBIS R
$
and 
$
\ctpp{\til Q}\Arcap{\mu}  \ctp{\til Q} 
. $
Finally, again applying the definition of $\mcontrBIS$ on 
$\ct{\til Q} \mcontrBIS \ctpp{\til Q}$, 
we derive 
$
\ct{\til Q}\Arcap{\mu}  \wb \ctp{\til Q} 
.$
\end{proof}

\begin{theorem}[unique solution of contractions for $\wb$]
\label{t:contraBisimulationU}
A system of weakly-guarded contractions
having only weakly-guarded sums, has a unique solution for $\wb$.
\end{theorem}

\begin{proof}{(sketch from \cite{sangiorgi2017equations})}
Suppose $\til P$ and $\til Q$ are two such solutions (for $\wb$) and consider
the relation
\begin{equation}
\label{eq:R}
\R \DSdefi \{ 
(R,S) \st R \wb \ct{\til P}, S \wb \ct{\til Q} \mbox{ for some context
$\qct$ (weakly-guarded sum only)} \} \enspace.
\end{equation}
We show that $\R$ is a bisimulation. \hl{Suppose $R\ \R\ S$ vis the context
$C$}, and $R \arr{\mu} R'$. We have to find $S'$ with $S \Arcap{\mu}
S'$ and $R'\ \R\ S'$. From $R \wb C[{\til P}]$, we derive $C[{\til P}]
\Arcap{\mu} R'' \wb R'$ for some $R''$. By Lemma~\ref{l:uptocon},
there is $C'$ with $R'' \mcontrBIS C'[{\til P}]$ and $C[{\til Q}]
\Arcap{\mu} \wb C'[{\til Q}]$. Hence, by definition of $\wb$, there is
also $S'$ with $S \Arcap{\mu} S' \wb C'[{\til Q}]$. This closes the
proof, as we have $R' \wb C'[{\til P}]$ and $S' \wb C'[{\til Q}]$.
\end{proof}

\subsection{Rooted contraction}
\label{ss:new}

The unique solution theorem of Section~\ref{ss:SysContr} requires a
constrained syntax for sums, due to the congruence and precongruence
problems of bisimilarity and contraction with such operator. 
We show here that the constraints can be
removed by moving to the induced congruence and precongruence, the
latter called \emph{rooted contraction}:
\begin{definition}
\label{d:rcontra}
Two processes $P$ and $Q$ are in \emph{rooted contraction}, written as
 $P\rcontr Q$, if
\begin{enumerate}
\item $P \arr\mu P'$ implies that there is $Q'$ with $Q \arr \mu Q'$
 and $P'\mcontrBIS Q'$;
\item $Q \arr\mu Q'$   implies that there is $P'$ with $P \Arr \mu
 P'$ and $P' \wb Q'$\enspace.
\end{enumerate}
\end{definition}


\hl{The precise formulation of  this definition was guided by the HOL theorem
  prover and
the following two principles:} (1) the definition should not be recursive,
along the lines of rooted bisimilarity
$\rapprox$ in Def.~\ref{d:rootedBisimilarity};
(2) the definition should  be built on top of existing \emph{contraction}
relation $\mcontrBIS$ (because of its completeness). 
\hl{A few other candidates were quickly tested and rejected, e.g.,
  because of  precongruence issue.} The proof of the precongruence
 result below is along the lines of the analogous result
for rooted bisimilarity with respect to bisimilarity.

\begin{theorem}
\label{t:rcontrPrecongruence}
$\rcontr$ is a precongruence in CCS, and it is the
coarsest precongruence contained in $\contr$.
\end{theorem}  

For a system of rooted contractions, the meaning of 
``solution for $\rcontr$'' and of \emph{a unique solution for $\rapprox$}
is the expected one --- just replace in Def.~\ref{d:uniContra}  the preorder 
$\contr$ with $\rcontr$, and the equivalence 
$\approx$ with $\rapprox$.
For this new relation, the analogous of Lemma~\ref{l:uptocon} and of
Theorem~\ref{t:contraBisimulationU} can now be stated without constraints on the sum
operator.
The schema of the proofs is almost identical, because all 
properties of $\rcontr$ needed in this proof is its precongruence, which is
indeed true on unrestricted contexts including direct sums:
\begin{lemma}
\label{l:ruptocon}
Suppose $\til P$ and $\til Q$ are solutions  for $\rcontr$ 
 of a system of weakly-guarded
contractions.
For any context $\qct$, 
if  $\ct{\til P}\Arr{\mu}  R$,
 then 
there is a  context $\qctp$
such that $R \mcontrBIS \ctp{\til P}$ and  $\ct{\til Q} \Arr{\mu}
 \wb \ctp{\til Q}$.
\end{lemma}

\begin{theorem}[unique solution of contractions for $\rapprox$]
\label{t:rcontraBisimulationU}
A system of weakly-guarded contractions has a unique solution 
 for $\rapprox$. (thus also for $\wb$)
\end{theorem} 

\begin{proof}
We first follow the same steps as in the proof of Theorem~\ref{t:contraBisimulationU} to show the relation $\R$ (now
with $\rcontr$ and unrestricted context $C$) in (\ref{eq:R}) is bisimulation,
exploting Lemma~\ref{l:ruptocon}. \hl{Then it remains to show that,} for
any two process $P$ and $Q$ with action $\mu$, if $P \arr{\mu} P'$ then
there is $Q'$ such that $Q \Arr{\mu} Q'$ (not $Q \Arcap{\mu} Q'$!) and
$P'\ \R\ Q'$, and also for the converse direction, exploting Lemma
4.13 of \cite{Mil89} \hl{(surprisingly)}. \hl{By definition of
\emph{bisimulation} (not $\wb$!) and $\approx^c$, we actually proved $P
\approx^c Q$ instead of $P \wb Q$.}
\end{proof}

\section{Formalisation}
\label{s:for}
We highlight here a formalisation of CCS
in the HOL theorem
prover (HOL4) \cite{slind2008brief},
including the new concepts and theorems proposed in the first half of
this paper.
The whole formalisation (apart from minor fixes and extensions in this
paper)
is described in \cite{Tian:2017wrba}, and the
proof scripts are in HOL4 official
examples\footnote{\url{https://github.com/HOL-Theorem-Prover/HOL/tree/master/examples/CCS}}. The
current work consists of about 20,000 lines of proof scripts in Standard ML.

Higher Order Logic (or \emph{HOL Logic}) \cite{hollogic}, which traces
its roots back to LCF
\cite{gordon1979edinburgh,milner1972logic} by Robin Milner and others since 1972, is a variant of
Church’s simple theory of types (STT) \cite{church1940formulation},
plus a higher order version of Hilbert's choice operator $\varepsilon$,
Axiom of Infinity, and Rank-1 (prenex) polymorphism.
HOL4 has implemented the original HOL Logic, 
while some other theorem provers in HOL family (e.g. Isabelle/HOL) have
certain extensions.
Indeed the HOL Logic has considerable simpler logical
foundations than most other theorem provers. 
As a consequence,
formal theories built in HOL is easily convincible and can
also be easily ported to other proof systems,
sometimes automatically \cite{hurd2011opentheory}.

HOL4 is written in Standard ML, a single programming language which
plays three different roles:
\begin{enumerate}
\item It serves as the underlying implementation language for the core HOL engine;\vspace{-1ex}
\item it is used to implement tactics (and tacticals) for writing proofs;\vspace{-1ex}
\item it is used as the command language of the HOL interactive shell.
\end{enumerate}
\hl{Moreover, using} the same language HOL4 users can write complex automatic
verification tools by calling HOL's theorem proving
facilities. \hl{(The formal proofs of theorems in CCS theory
are mostly done by an \emph{interactive process} closely following
their informal proofs, with minimal automatic proof searching.)}


\emph{In this
formalisation we consider only single-variable equations/contractions.}
This considerably
 simplifies the required proofs in HOL, also enhances the readability of
 proof scripts \hl{without loss of generality. (For paper proofs, the
 multi-variable case is just a routine adaptation.)}

\subsection{CCS and its transitions by SOS rules}

In our CCS formalisation, the type ``\HOLinline{\ensuremath{\beta} \HOLTyOp{Label}}'' (\texttt{'b} or
$\beta$ is the type variable for actions) accounts for visible actions, divided into input
and output actions, defined by HOL's Datatype package:
\begin{lstlisting}
val _ = Datatype `Label = name 'b | coname 'b`;
\end{lstlisting}
The type ``\HOLinline{\ensuremath{\beta} \HOLTyOp{Action}}'' is the
union of all visible actions, plus invisible action $\tau$ (now based on
HOL's \texttt{option} theory). The cardinality of
``\HOLinline{\ensuremath{\beta} \HOLTyOp{Action}}'' (and therefore of all
CCS types built on top of it)
 depends on the choice (or \emph{type-instantiation}) of $\beta$.

The type ``\HOLinline{(\ensuremath{\alpha}, \ensuremath{\beta}) \HOLTyOp{CCS}}'', accounting for the CCS
syntax\footnote{\hl{The order of type variables $\alpha$ and $\beta$
    is irrelevant. Our choice is aligned with other CCS literals.
$\mathrm{CCS}(h,k)$ is the CCS subcalculus which can use at most $h$ constants
and $k$ actions.} \cite{gorrieri2017ccs} \hl{Thus, to formalize theorems on
such a CCS subcalculus, the needed CCS type can be retrieved by instantiating the type
variables $\alpha$ and $\beta$ in ``\HOLinline{(\ensuremath{\alpha}, \ensuremath{\beta}) \HOLTyOp{CCS}}'' with types
having the corresponding cardinalities $h$ and $k$. Monica Nesi goes
too far by adding another type variable $\gamma$ for value-passing CCS \cite{Nesi:2017wo}.}}, is then defined inductively:
(\texttt{'a} or $\alpha$ is the type variable for recursion variables,
``\HOLinline{\ensuremath{\beta} \HOLTyOp{Relabeling}}'' is \hl{the type of all relabeling functions,
\mbox{\color{blue}{\texttt{`}}} is for backquotes of HOL terms}):
\begin{lstlisting}
val _ = Datatype `CCS = nil
		      | var 'a
		      | prefix ('b Action) CCS
		      | sum CCS CCS
		      | par CCS CCS
		      | restr (('b Label) set) CCS
		      | relab CCS ('b Relabeling)
		      | rec 'a CCS`;
\end{lstlisting}

We have added some grammar support,
 using \hl{HOL's powerful pretty printer}, to represent CCS
processes in more readable forms (c.f. the column \textbf{HOL (abbrev.)}
in Table \ref{tab:ccsoperator}, \hl{which} summarizes 
the main syntactic notations \hl{of} CCS). For the restriction
operator, we have chosen to allow a  set of names as a parameter, rather than a
  single name as in the ordinary  CCS syntax; this simplifies 
the manipulation of 
 processes with different orders of
  nested restrictions.

\begin{table}[h]
\begin{center}
\begin{tabular}{|c|c|c|c|}
\hline
\textbf{Operator} & \textbf{CCS Notation} & \textbf{HOL term} &
                                                                \textbf{HOL (abbrev.)}\\
\hline
nil & $\textbf{0}$ & \HOLinline{\HOLConst{nil}} & \HOLinline{\HOLConst{nil}} \\
prefix & $u.P$ & \texttt{prefix u P} & \HOLinline{\HOLFreeVar{u}\HOLSymConst{..}\HOLFreeVar{P}} \\
sum & $P + Q$ & \texttt{sum P Q} & \HOLinline{\HOLFreeVar{P} \HOLSymConst{\ensuremath{+}} \HOLFreeVar{Q}} \\
parallel & $P \,\mid\, Q$ & \texttt{par P Q} & \HOLinline{\HOLFreeVar{P} \HOLSymConst{\ensuremath{\parallel}} \HOLFreeVar{Q}} \\
restriction & $(\nu\;L)\;P$ & \texttt{restr L P} & \HOLinline{\HOLSymConst{\ensuremath{\nu}} \HOLFreeVar{L} \HOLFreeVar{P}}
  \\
recursion & $\recu A P$ & \texttt{rec A P} & \HOLinline{\HOLConst{rec} \HOLFreeVar{A} \HOLFreeVar{P}}
  \\
relabeling & $P\;[r\!\!f]$ & \texttt{relab P rf} & \HOLinline{\HOLConst{relab} \HOLFreeVar{P} \HOLFreeVar{rf}}
  \\
\hline
constant & $A$ & \texttt{var A} & \HOLinline{\HOLConst{var} \HOLFreeVar{A}} \\
invisible action & $\tau$ & \texttt{tau} & \HOLinline{\HOLSymConst{\ensuremath{\tau}}} \\
input action & $a$ & \texttt{label (name a)} & \HOLinline{\HOLConst{In} \HOLFreeVar{a}} \\
output action & $\outC a$ & \texttt{label (coname a)} & \HOLinline{\HOLConst{Out} \HOLFreeVar{a}} \\
\hline
\end{tabular}
\end{center}
\vspace{-1em}
   \caption{Syntax of CCS operators, constant and actions}
   \label{tab:ccsoperator}
\end{table}

The transition semantics of CCS processes follows Structural
Operational Semantics (SOS) in Fig.~\ref{f:LTSCCS}:
\begin{alltt}
\HOLTokenTurnstile{} \HOLFreeVar{u}\HOLSymConst{..}\HOLFreeVar{P} \HOLTokenTransBegin\HOLFreeVar{u}\HOLTokenTransEnd \HOLFreeVar{P}\hfill\texttt{[PREFIX]}
\HOLTokenTurnstile{} \HOLFreeVar{P} \HOLTokenTransBegin\HOLFreeVar{u}\HOLTokenTransEnd \HOLFreeVar{P\sp{\prime}} \HOLSymConst{\HOLTokenImp{}} \HOLFreeVar{P} \HOLSymConst{\ensuremath{+}} \HOLFreeVar{Q} \HOLTokenTransBegin\HOLFreeVar{u}\HOLTokenTransEnd \HOLFreeVar{P\sp{\prime}}\hfill\texttt{[SUM1]}
\HOLTokenTurnstile{} \HOLFreeVar{P} \HOLTokenTransBegin\HOLFreeVar{u}\HOLTokenTransEnd \HOLFreeVar{P\sp{\prime}} \HOLSymConst{\HOLTokenImp{}} \HOLFreeVar{Q} \HOLSymConst{\ensuremath{+}} \HOLFreeVar{P} \HOLTokenTransBegin\HOLFreeVar{u}\HOLTokenTransEnd \HOLFreeVar{P\sp{\prime}}\hfill\texttt{[SUM2]}
\HOLTokenTurnstile{} \HOLFreeVar{P} \HOLTokenTransBegin\HOLFreeVar{u}\HOLTokenTransEnd \HOLFreeVar{P\sp{\prime}} \HOLSymConst{\HOLTokenImp{}} \HOLFreeVar{P} \HOLSymConst{\ensuremath{\parallel}} \HOLFreeVar{Q} \HOLTokenTransBegin\HOLFreeVar{u}\HOLTokenTransEnd \HOLFreeVar{P\sp{\prime}} \HOLSymConst{\ensuremath{\parallel}} \HOLFreeVar{Q}\hfill\texttt{[PAR1]}
\HOLTokenTurnstile{} \HOLFreeVar{P} \HOLTokenTransBegin\HOLFreeVar{u}\HOLTokenTransEnd \HOLFreeVar{P\sp{\prime}} \HOLSymConst{\HOLTokenImp{}} \HOLFreeVar{Q} \HOLSymConst{\ensuremath{\parallel}} \HOLFreeVar{P} \HOLTokenTransBegin\HOLFreeVar{u}\HOLTokenTransEnd \HOLFreeVar{Q} \HOLSymConst{\ensuremath{\parallel}} \HOLFreeVar{P\sp{\prime}}\hfill\texttt{[PAR2]}
\HOLTokenTurnstile{} \HOLFreeVar{P} \HOLTokenTransBegin\HOLConst{label} \HOLFreeVar{l}\HOLTokenTransEnd \HOLFreeVar{P\sp{\prime}} \HOLSymConst{\HOLTokenConj{}} \HOLFreeVar{Q} \HOLTokenTransBegin\HOLConst{label} (\HOLConst{COMPL} \HOLFreeVar{l})\HOLTokenTransEnd \HOLFreeVar{Q\sp{\prime}} \HOLSymConst{\HOLTokenImp{}} \HOLFreeVar{P} \HOLSymConst{\ensuremath{\parallel}} \HOLFreeVar{Q} \HOLTokenTransBegin\HOLSymConst{\ensuremath{\tau}}\HOLTokenTransEnd \HOLFreeVar{P\sp{\prime}} \HOLSymConst{\ensuremath{\parallel}} \HOLFreeVar{Q\sp{\prime}}\hfill\texttt{[PAR3]}
\HOLTokenTurnstile{} \HOLFreeVar{P} \HOLTokenTransBegin\HOLFreeVar{u}\HOLTokenTransEnd \HOLFreeVar{Q} \HOLSymConst{\HOLTokenConj{}} ((\HOLFreeVar{u} \HOLSymConst{=} \HOLSymConst{\ensuremath{\tau}}) \HOLSymConst{\HOLTokenDisj{}} (\HOLFreeVar{u} \HOLSymConst{=} \HOLConst{label} \HOLFreeVar{l}) \HOLSymConst{\HOLTokenConj{}} \HOLFreeVar{l} \HOLSymConst{\HOLTokenNotIn{}} \HOLFreeVar{L} \HOLSymConst{\HOLTokenConj{}} \HOLConst{COMPL} \HOLFreeVar{l} \HOLSymConst{\HOLTokenNotIn{}} \HOLFreeVar{L}) \HOLSymConst{\HOLTokenImp{}}
   \HOLSymConst{\ensuremath{\nu}} \HOLFreeVar{L} \HOLFreeVar{P} \HOLTokenTransBegin\HOLFreeVar{u}\HOLTokenTransEnd \HOLSymConst{\ensuremath{\nu}} \HOLFreeVar{L} \HOLFreeVar{Q}\hfill\texttt{[RESTR]}
\HOLTokenTurnstile{} \HOLFreeVar{P} \HOLTokenTransBegin\HOLFreeVar{u}\HOLTokenTransEnd \HOLFreeVar{Q} \HOLSymConst{\HOLTokenImp{}} \HOLConst{relab} \HOLFreeVar{P} \HOLFreeVar{rf} \HOLTokenTransBegin\HOLConst{relabel} \HOLFreeVar{rf} \HOLFreeVar{u}\HOLTokenTransEnd \HOLConst{relab} \HOLFreeVar{Q} \HOLFreeVar{rf}\hfill\texttt{[RELABELING]}
\HOLTokenTurnstile{} \HOLConst{CCS_Subst} \HOLFreeVar{P} (\HOLConst{rec} \HOLFreeVar{A} \HOLFreeVar{P}) \HOLFreeVar{A} \HOLTokenTransBegin\HOLFreeVar{u}\HOLTokenTransEnd \HOLFreeVar{P\sp{\prime}} \HOLSymConst{\HOLTokenImp{}} \HOLConst{rec} \HOLFreeVar{A} \HOLFreeVar{P} \HOLTokenTransBegin\HOLFreeVar{u}\HOLTokenTransEnd \HOLFreeVar{P\sp{\prime}}\hfill\texttt{[REC]}
\end{alltt}

The rule \texttt{REC} (Recursion)
 says that if we substitute all appearances of variable $A$ in $P$ to
$(\recu A P)$ and the resulting process has a transition to $P'$
with action $u$, then $(\recu A P)$ has the same
transition. From HOL's viewpoint, these
SOS rules are \emph{inductive 
  definitions} on the tenary relation \HOLinline{\HOLConst{TRANS}} of type ``\HOLinline{(\ensuremath{\alpha}, \ensuremath{\beta}) \HOLTyOp{CCS} \HOLTokenTransEnd \ensuremath{\beta} \HOLTyOp{Action} \HOLTokenTransEnd (\ensuremath{\alpha}, \ensuremath{\beta}) \HOLTyOp{CCS} \HOLTokenTransEnd \HOLTyOp{bool}}'', generated by HOL's 
\texttt{Hol_reln} function.

A useful function that we have defined, exploiting the interplay
between HOL4 and Standard ML (and following an idea by Nesi \cite{Nesi:1992ve})
 is a complex Standard ML function
  taking a CCS process and returning a theorem indicating all its
  direct transitions.\footnote{If the input process could yield
    something infinite branching, due to the use of recursion or
    relabeling operators, the program will loop forever without
    outputting a theorem.}
For instance, we know that the process $(a.0 | \bar{a}.0)$ has three
possible transitions: $(a.0 | \bar{a}.0) \overset{a}{\longrightarrow}
(0 | \bar{a}.0)$, $(a.0 | \bar{a}.0)
\overset{\bar{a}}{\longrightarrow} (a.0 | 0)$ and $(a.0 | \bar{a}.0)
\overset{\tau}{\longrightarrow} (0 | 0)$.
To completely describe all possible transitions of a process, if done manually, the
following facts should be proved: (1) there exists transitions from
$(a.0 | \bar{a}.0)$ (optional); (2) the correctness for each of the
transitions; and (3) the non-existence of other transitions.

For large processes it may be surprisingly hard to manually prove the
non-existence of transitions.  Hence the usefulness of appealing to 
the new  function \texttt{CCS\_TRANS\_CONV}. 
For instance this function
is called on the  process $(a.0 | \bar{a}.0)$ thus:
\hl{(\mbox{\color{blue}{\texttt{``}}} is for double-backquotes of HOL
  terms, \mbox{\color{blue}{\texttt{>}}} is HOL's prompt)}
\begin{lstlisting}
> CCS_TRANS_CONV ``par (prefix (label (name "a")) nil)
                       (prefix (label (coname "a")) nil)``
\end{lstlisting}
This returns the following theorem, indeed describing all immediate
transitions of the process:
\begin{alltt}
\HOLTokenTurnstile{} \HOLConst{In} \HOLStringLit{a}\HOLSymConst{..}\HOLConst{nil} \HOLSymConst{\ensuremath{\parallel}} \HOLConst{Out} \HOLStringLit{a}\HOLSymConst{..}\HOLConst{nil} \HOLTokenTransBegin\HOLFreeVar{u}\HOLTokenTransEnd \HOLFreeVar{E} \HOLSymConst{\HOLTokenEquiv{}}
   ((\HOLFreeVar{u} \HOLSymConst{=} \HOLConst{In} \HOLStringLit{a}) \HOLSymConst{\HOLTokenConj{}} (\HOLFreeVar{E} \HOLSymConst{=} \HOLConst{nil} \HOLSymConst{\ensuremath{\parallel}} \HOLConst{Out} \HOLStringLit{a}\HOLSymConst{..}\HOLConst{nil}) \HOLSymConst{\HOLTokenDisj{}}
    (\HOLFreeVar{u} \HOLSymConst{=} \HOLConst{Out} \HOLStringLit{a}) \HOLSymConst{\HOLTokenConj{}} (\HOLFreeVar{E} \HOLSymConst{=} \HOLConst{In} \HOLStringLit{a}\HOLSymConst{..}\HOLConst{nil} \HOLSymConst{\ensuremath{\parallel}} \HOLConst{nil})) \HOLSymConst{\HOLTokenDisj{}}
   (\HOLFreeVar{u} \HOLSymConst{=} \HOLSymConst{\ensuremath{\tau}}) \HOLSymConst{\HOLTokenConj{}} (\HOLFreeVar{E} \HOLSymConst{=} \HOLConst{nil} \HOLSymConst{\ensuremath{\parallel}} \HOLConst{nil})\hfill{[Example.ex_A]}
\end{alltt}


\subsection{Bisimulation and Bisimilarity}

\hl{To define (weak) bisimilarity, we first need to define weak
transitions of CCS processes. Following the name adopted by Nesi} \cite{Nesi:1992ve},
we define a (possibly empty) sequence of $\tau$-transitions between
two processes as
a new relation called \texttt{EPS} ($\overset{\epsilon}{\Rightarrow}$), which is the RTC
(reflexive transitive closure, denoted by \mbox{\color{blue}{$^*$}} in
HOL4) of ordinary $\tau$-transitions of CCS processes:
\begin{alltt}
\HOLConst{EPS} \HOLSymConst{=} (\HOLTokenLambda{}\HOLBoundVar{E} \HOLBoundVar{E\sp{\prime}}. \HOLBoundVar{E} \HOLTokenTransBegin\HOLSymConst{\ensuremath{\tau}}\HOLTokenTransEnd \HOLBoundVar{E\sp{\prime}})\HOLSymConst{\HOLTokenSupStar{}}\hfill{[EPS_def]}
\end{alltt}
Then we can define a weak transition as an ordinary transition wrapped by
two $\epsilon$-transitions:
\begin{alltt}
\HOLFreeVar{E} \HOLTokenWeakTransBegin\HOLFreeVar{u}\HOLTokenWeakTransEnd \HOLFreeVar{E\sp{\prime}} \HOLSymConst{\HOLTokenEquiv{}} \HOLSymConst{\HOLTokenExists{}}\HOLBoundVar{E\sb{\mathrm{1}}} \HOLBoundVar{E\sb{\mathrm{2}}}. \HOLFreeVar{E} \HOLSymConst{\HOLTokenEPS} \HOLBoundVar{E\sb{\mathrm{1}}} \HOLSymConst{\HOLTokenConj{}} \HOLBoundVar{E\sb{\mathrm{1}}} \HOLTokenTransBegin\HOLFreeVar{u}\HOLTokenTransEnd \HOLBoundVar{E\sb{\mathrm{2}}} \HOLSymConst{\HOLTokenConj{}} \HOLBoundVar{E\sb{\mathrm{2}}} \HOLSymConst{\HOLTokenEPS} \HOLFreeVar{E\sp{\prime}}\hfill{[WEAK_TRANS]}
\end{alltt}

For the definition of bisimilarity and the associated coinduction
principle \cite{sangiorgi2011advanced}, we have taken
advantage of HOL's coinductive relation package (\texttt{Hol_coreln} \cite{holdesc}),
a new tool since its Kananaskis-11 release (March 3,
2017).\footnote{\url{https://hol-theorem-prover.org/kananaskis-11.release.html\#new-tools}}
This essentially amounts to defining bisimilarity as the greatest
fixed-point of the appropriate functional on relations. 
Precisely we call 
the \texttt{Hol_coreln}
command as follows: (here \texttt{WB} is meant to be
\texttt{WEAK_EQUIV} ($\approx$) in the rest of this paper;
{\tt !} and {\tt ?} stand for universal and
existential quantifiers.)
\begin{lstlisting}
val (WB_rules, WB_coind, WB_cases) = Hol_coreln `
    (!(P :('a, 'b) CCS) (Q :('a, 'b) CCS).
       (!l.
	 (!P'. TRANS P (label l) P' ==>
	       (?Q'. WEAK_TRANS Q (label l) Q' /\ WB P' Q')) /\
	 (!Q'. TRANS Q (label l) Q' ==>
	       (?P'. WEAK_TRANS P (label l) P' /\ WB P' Q'))) /\
       (!P'. TRANS P tau P' ==> (?Q'. EPS Q Q' /\ WB P' Q')) /\
       (!Q'. TRANS Q tau Q' ==> (?P'. EPS P P' /\ WB P' Q'))
     ==> WB P Q)`;
\end{lstlisting}
\texttt{Hol_coreln} returns 3 theorems, of the first being always the
same as input term\footnote{\hl{Our mixing of HOL notation and mathematical
  notation in this paper is not arbitrary. We have to paste here the
  original proof scripts, which is written in HOL's ASCII term
  notation} (c.f. \cite{holdesc} for more details). \hl{HOL4 also supports writing Unicode symbols directly in
  proof scripts but we did not make use of them. However, all formal definitions and
  theorems in the paper are automatically generated from HOL4 in
  which we have made an effort for generating
  Unicode and TeX outputs as natural as possible. What is really
  arbitrary is the presense/absense of outermost universal
  quantifiers in all generated theorems.}} (now proved automatically as a theorem).
The second and third theorems, namely \texttt{WB_coind} and \texttt{WB_cases},
express the coinduction proof method for bisimilarity 
(i.e.~any bisimulation is contained in bisimilarity)
and the fixed-point property of bisimilarity
(bisimilarity itself is a bisimulation, thus the largest
bisimulation):
\begin{enumerate}
\item 
\begin{small}
\begin{alltt}
\HOLTokenTurnstile{} \HOLSymConst{\HOLTokenForall{}}\HOLBoundVar{WB\sp{\prime}}.
       (\HOLSymConst{\HOLTokenForall{}}\HOLBoundVar{a\sb{\mathrm{0}}} \HOLBoundVar{a\sb{\mathrm{1}}}.
            \HOLBoundVar{WB\sp{\prime}} \HOLBoundVar{a\sb{\mathrm{0}}} \HOLBoundVar{a\sb{\mathrm{1}}} \HOLSymConst{\HOLTokenImp{}}
            (\HOLSymConst{\HOLTokenForall{}}\HOLBoundVar{l}.
                 (\HOLSymConst{\HOLTokenForall{}}\HOLBoundVar{P\sp{\prime}}.
                      \HOLBoundVar{a\sb{\mathrm{0}}} \HOLTokenTransBegin\HOLConst{label} \HOLBoundVar{l}\HOLTokenTransEnd \HOLBoundVar{P\sp{\prime}} \HOLSymConst{\HOLTokenImp{}}
                      \HOLSymConst{\HOLTokenExists{}}\HOLBoundVar{Q\sp{\prime}}. \HOLBoundVar{a\sb{\mathrm{1}}} \HOLTokenWeakTransBegin\HOLConst{label} \HOLBoundVar{l}\HOLTokenWeakTransEnd \HOLBoundVar{Q\sp{\prime}} \HOLSymConst{\HOLTokenConj{}} \HOLBoundVar{WB\sp{\prime}} \HOLBoundVar{P\sp{\prime}} \HOLBoundVar{Q\sp{\prime}}) \HOLSymConst{\HOLTokenConj{}}
                 \HOLSymConst{\HOLTokenForall{}}\HOLBoundVar{Q\sp{\prime}}.
                     \HOLBoundVar{a\sb{\mathrm{1}}} \HOLTokenTransBegin\HOLConst{label} \HOLBoundVar{l}\HOLTokenTransEnd \HOLBoundVar{Q\sp{\prime}} \HOLSymConst{\HOLTokenImp{}}
                     \HOLSymConst{\HOLTokenExists{}}\HOLBoundVar{P\sp{\prime}}. \HOLBoundVar{a\sb{\mathrm{0}}} \HOLTokenWeakTransBegin\HOLConst{label} \HOLBoundVar{l}\HOLTokenWeakTransEnd \HOLBoundVar{P\sp{\prime}} \HOLSymConst{\HOLTokenConj{}} \HOLBoundVar{WB\sp{\prime}} \HOLBoundVar{P\sp{\prime}} \HOLBoundVar{Q\sp{\prime}}) \HOLSymConst{\HOLTokenConj{}}
            (\HOLSymConst{\HOLTokenForall{}}\HOLBoundVar{P\sp{\prime}}. \HOLBoundVar{a\sb{\mathrm{0}}} \HOLTokenTransBegin\HOLSymConst{\ensuremath{\tau}}\HOLTokenTransEnd \HOLBoundVar{P\sp{\prime}} \HOLSymConst{\HOLTokenImp{}} \HOLSymConst{\HOLTokenExists{}}\HOLBoundVar{Q\sp{\prime}}. \HOLBoundVar{a\sb{\mathrm{1}}} \HOLSymConst{\HOLTokenEPS} \HOLBoundVar{Q\sp{\prime}} \HOLSymConst{\HOLTokenConj{}} \HOLBoundVar{WB\sp{\prime}} \HOLBoundVar{P\sp{\prime}} \HOLBoundVar{Q\sp{\prime}}) \HOLSymConst{\HOLTokenConj{}}
            \HOLSymConst{\HOLTokenForall{}}\HOLBoundVar{Q\sp{\prime}}. \HOLBoundVar{a\sb{\mathrm{1}}} \HOLTokenTransBegin\HOLSymConst{\ensuremath{\tau}}\HOLTokenTransEnd \HOLBoundVar{Q\sp{\prime}} \HOLSymConst{\HOLTokenImp{}} \HOLSymConst{\HOLTokenExists{}}\HOLBoundVar{P\sp{\prime}}. \HOLBoundVar{a\sb{\mathrm{0}}} \HOLSymConst{\HOLTokenEPS} \HOLBoundVar{P\sp{\prime}} \HOLSymConst{\HOLTokenConj{}} \HOLBoundVar{WB\sp{\prime}} \HOLBoundVar{P\sp{\prime}} \HOLBoundVar{Q\sp{\prime}}) \HOLSymConst{\HOLTokenImp{}}
       \HOLSymConst{\HOLTokenForall{}}\HOLBoundVar{a\sb{\mathrm{0}}} \HOLBoundVar{a\sb{\mathrm{1}}}. \HOLBoundVar{WB\sp{\prime}} \HOLBoundVar{a\sb{\mathrm{0}}} \HOLBoundVar{a\sb{\mathrm{1}}} \HOLSymConst{\HOLTokenImp{}} \HOLConst{WB} \HOLBoundVar{a\sb{\mathrm{0}}} \HOLBoundVar{a\sb{\mathrm{1}}}\hfill{[WB_coind, WEAK_EQUIV_coind]}
\end{alltt}
\end{small}
\item 
\begin{small}
\begin{alltt}
\HOLTokenTurnstile{} \HOLSymConst{\HOLTokenForall{}}\HOLBoundVar{a\sb{\mathrm{0}}} \HOLBoundVar{a\sb{\mathrm{1}}}.
       \HOLConst{WB} \HOLBoundVar{a\sb{\mathrm{0}}} \HOLBoundVar{a\sb{\mathrm{1}}} \HOLSymConst{\HOLTokenEquiv{}}
       (\HOLSymConst{\HOLTokenForall{}}\HOLBoundVar{l}.
            (\HOLSymConst{\HOLTokenForall{}}\HOLBoundVar{P\sp{\prime}}.
                 \HOLBoundVar{a\sb{\mathrm{0}}} \HOLTokenTransBegin\HOLConst{label} \HOLBoundVar{l}\HOLTokenTransEnd \HOLBoundVar{P\sp{\prime}} \HOLSymConst{\HOLTokenImp{}}
                 \HOLSymConst{\HOLTokenExists{}}\HOLBoundVar{Q\sp{\prime}}. \HOLBoundVar{a\sb{\mathrm{1}}} \HOLTokenWeakTransBegin\HOLConst{label} \HOLBoundVar{l}\HOLTokenWeakTransEnd \HOLBoundVar{Q\sp{\prime}} \HOLSymConst{\HOLTokenConj{}} \HOLConst{WB} \HOLBoundVar{P\sp{\prime}} \HOLBoundVar{Q\sp{\prime}}) \HOLSymConst{\HOLTokenConj{}}
            \HOLSymConst{\HOLTokenForall{}}\HOLBoundVar{Q\sp{\prime}}.
                \HOLBoundVar{a\sb{\mathrm{1}}} \HOLTokenTransBegin\HOLConst{label} \HOLBoundVar{l}\HOLTokenTransEnd \HOLBoundVar{Q\sp{\prime}} \HOLSymConst{\HOLTokenImp{}}
                \HOLSymConst{\HOLTokenExists{}}\HOLBoundVar{P\sp{\prime}}. \HOLBoundVar{a\sb{\mathrm{0}}} \HOLTokenWeakTransBegin\HOLConst{label} \HOLBoundVar{l}\HOLTokenWeakTransEnd \HOLBoundVar{P\sp{\prime}} \HOLSymConst{\HOLTokenConj{}} \HOLConst{WB} \HOLBoundVar{P\sp{\prime}} \HOLBoundVar{Q\sp{\prime}}) \HOLSymConst{\HOLTokenConj{}}
       (\HOLSymConst{\HOLTokenForall{}}\HOLBoundVar{P\sp{\prime}}. \HOLBoundVar{a\sb{\mathrm{0}}} \HOLTokenTransBegin\HOLSymConst{\ensuremath{\tau}}\HOLTokenTransEnd \HOLBoundVar{P\sp{\prime}} \HOLSymConst{\HOLTokenImp{}} \HOLSymConst{\HOLTokenExists{}}\HOLBoundVar{Q\sp{\prime}}. \HOLBoundVar{a\sb{\mathrm{1}}} \HOLSymConst{\HOLTokenEPS} \HOLBoundVar{Q\sp{\prime}} \HOLSymConst{\HOLTokenConj{}} \HOLConst{WB} \HOLBoundVar{P\sp{\prime}} \HOLBoundVar{Q\sp{\prime}}) \HOLSymConst{\HOLTokenConj{}}
       \HOLSymConst{\HOLTokenForall{}}\HOLBoundVar{Q\sp{\prime}}. \HOLBoundVar{a\sb{\mathrm{1}}} \HOLTokenTransBegin\HOLSymConst{\ensuremath{\tau}}\HOLTokenTransEnd \HOLBoundVar{Q\sp{\prime}} \HOLSymConst{\HOLTokenImp{}} \HOLSymConst{\HOLTokenExists{}}\HOLBoundVar{P\sp{\prime}}. \HOLBoundVar{a\sb{\mathrm{0}}} \HOLSymConst{\HOLTokenEPS} \HOLBoundVar{P\sp{\prime}} \HOLSymConst{\HOLTokenConj{}} \HOLConst{WB} \HOLBoundVar{P\sp{\prime}} \HOLBoundVar{Q\sp{\prime}}\hfill{[WB_cases, WEAK_EQUIV_cases]}
\end{alltt}
\end{small}
\end{enumerate}

\hl{The coinduction principle \texttt{WB_coind} says that any
bisimulation is contained in the resulting relation (i.e.~it is
largest), but it didn't constrain the resulting relation in the set of
fixed points (e.g.~even the universal relation --- the set of all
pairs --- would fit with this theorem); the
purpose of \texttt{WB_cases} is to
further assert that the resulting relation is indeed a
fixed point. Thus \texttt{WB_coind} and \texttt{WB_cases}
together make sure that bisimilarity is the greatest
fixed point, as
the former contributes to ``greatest'' while the latter
contributes to ``fixed point''.}
Without HOL's coinductive relation package, bisimilarity
would have to be defined by following literally
Def.~\ref{d:wb};  then other properties of bisimilarity, such
as the fixed-point property in \texttt{WB_cases}, would have to be
derived manually (which is quite hard; indeed it was one of the main results
in Nesi's formalisation work  in HOL88 \cite{Nesi:1992ve}).



\subsection{Context, guardedness and (pre)congruence}


We have chosen to use $\lambda$-expressions (\hl{with the type}
``\HOLinline{(\ensuremath{\alpha}, \ensuremath{\beta}) \HOLTyOp{CCS} \HOLTokenTransEnd (\ensuremath{\alpha}, \ensuremath{\beta}) \HOLTyOp{CCS}}'')
to represent \emph{multi-hole contexts}.
\hl{This choice has a significant advantage over \emph{one-hole
contexts}, as each hole corresponds to one
appearance of the \emph{same} variable in single-variable
expressions (or equations). Thus \emph{contexts} can be directly used in
formulating the unique solution of equations theorems in
single-variable cases.} The precise definition is given inductively:
\begin{alltt}
\HOLConst{CONTEXT} (\HOLTokenLambda{}\HOLBoundVar{t}. \HOLBoundVar{t})
\HOLConst{CONTEXT} (\HOLTokenLambda{}\HOLBoundVar{t}. \HOLFreeVar{p})
\HOLConst{CONTEXT} \HOLFreeVar{e} \HOLSymConst{\HOLTokenImp{}} \HOLConst{CONTEXT} (\HOLTokenLambda{}\HOLBoundVar{t}. \HOLFreeVar{a}\HOLSymConst{..}\HOLFreeVar{e} \HOLBoundVar{t})
\HOLConst{CONTEXT} \HOLFreeVar{e\sb{\mathrm{1}}} \HOLSymConst{\HOLTokenConj{}} \HOLConst{CONTEXT} \HOLFreeVar{e\sb{\mathrm{2}}} \HOLSymConst{\HOLTokenImp{}} \HOLConst{CONTEXT} (\HOLTokenLambda{}\HOLBoundVar{t}. \HOLFreeVar{e\sb{\mathrm{1}}} \HOLBoundVar{t} \HOLSymConst{\ensuremath{+}} \HOLFreeVar{e\sb{\mathrm{2}}} \HOLBoundVar{t})
\HOLConst{CONTEXT} \HOLFreeVar{e\sb{\mathrm{1}}} \HOLSymConst{\HOLTokenConj{}} \HOLConst{CONTEXT} \HOLFreeVar{e\sb{\mathrm{2}}} \HOLSymConst{\HOLTokenImp{}} \HOLConst{CONTEXT} (\HOLTokenLambda{}\HOLBoundVar{t}. \HOLFreeVar{e\sb{\mathrm{1}}} \HOLBoundVar{t} \HOLSymConst{\ensuremath{\parallel}} \HOLFreeVar{e\sb{\mathrm{2}}} \HOLBoundVar{t})
\HOLConst{CONTEXT} \HOLFreeVar{e} \HOLSymConst{\HOLTokenImp{}} \HOLConst{CONTEXT} (\HOLTokenLambda{}\HOLBoundVar{t}. \HOLSymConst{\ensuremath{\nu}} \HOLFreeVar{L} (\HOLFreeVar{e} \HOLBoundVar{t}))
\HOLConst{CONTEXT} \HOLFreeVar{e} \HOLSymConst{\HOLTokenImp{}} \HOLConst{CONTEXT} (\HOLTokenLambda{}\HOLBoundVar{t}. \HOLConst{relab} (\HOLFreeVar{e} \HOLBoundVar{t}) \HOLFreeVar{rf})\hfill{[CONTEXT_rules]}
\end{alltt}

A context is \emph{weakly guarded} (\texttt{WG}) if each hole is
underneath a prefix:
\begin{alltt}
\HOLConst{WG} (\HOLTokenLambda{}\HOLBoundVar{t}. \HOLFreeVar{p})
\HOLConst{CONTEXT} \HOLFreeVar{e} \HOLSymConst{\HOLTokenImp{}} \HOLConst{WG} (\HOLTokenLambda{}\HOLBoundVar{t}. \HOLFreeVar{a}\HOLSymConst{..}\HOLFreeVar{e} \HOLBoundVar{t})
\HOLConst{WG} \HOLFreeVar{e\sb{\mathrm{1}}} \HOLSymConst{\HOLTokenConj{}} \HOLConst{WG} \HOLFreeVar{e\sb{\mathrm{2}}} \HOLSymConst{\HOLTokenImp{}} \HOLConst{WG} (\HOLTokenLambda{}\HOLBoundVar{t}. \HOLFreeVar{e\sb{\mathrm{1}}} \HOLBoundVar{t} \HOLSymConst{\ensuremath{+}} \HOLFreeVar{e\sb{\mathrm{2}}} \HOLBoundVar{t})
\HOLConst{WG} \HOLFreeVar{e\sb{\mathrm{1}}} \HOLSymConst{\HOLTokenConj{}} \HOLConst{WG} \HOLFreeVar{e\sb{\mathrm{2}}} \HOLSymConst{\HOLTokenImp{}} \HOLConst{WG} (\HOLTokenLambda{}\HOLBoundVar{t}. \HOLFreeVar{e\sb{\mathrm{1}}} \HOLBoundVar{t} \HOLSymConst{\ensuremath{\parallel}} \HOLFreeVar{e\sb{\mathrm{2}}} \HOLBoundVar{t})
\HOLConst{WG} \HOLFreeVar{e} \HOLSymConst{\HOLTokenImp{}} \HOLConst{WG} (\HOLTokenLambda{}\HOLBoundVar{t}. \HOLSymConst{\ensuremath{\nu}} \HOLFreeVar{L} (\HOLFreeVar{e} \HOLBoundVar{t}))
\HOLConst{WG} \HOLFreeVar{e} \HOLSymConst{\HOLTokenImp{}} \HOLConst{WG} (\HOLTokenLambda{}\HOLBoundVar{t}. \HOLConst{relab} (\HOLFreeVar{e} \HOLBoundVar{t}) \HOLFreeVar{rf})\hfill{[WG_rules]}
\end{alltt}

A context is \emph{(strongly) guarded} (\texttt{SG}) if each hole is underneath a \emph{visible} prefix:
\begin{alltt}
\HOLConst{SG} (\HOLTokenLambda{}\HOLBoundVar{t}. \HOLFreeVar{p})
\HOLConst{CONTEXT} \HOLFreeVar{e} \HOLSymConst{\HOLTokenImp{}} \HOLConst{SG} (\HOLTokenLambda{}\HOLBoundVar{t}. \HOLConst{label} \HOLFreeVar{l}\HOLSymConst{..}\HOLFreeVar{e} \HOLBoundVar{t})
\HOLConst{SG} \HOLFreeVar{e} \HOLSymConst{\HOLTokenImp{}} \HOLConst{SG} (\HOLTokenLambda{}\HOLBoundVar{t}. \HOLFreeVar{a}\HOLSymConst{..}\HOLFreeVar{e} \HOLBoundVar{t})
\HOLConst{SG} \HOLFreeVar{e\sb{\mathrm{1}}} \HOLSymConst{\HOLTokenConj{}} \HOLConst{SG} \HOLFreeVar{e\sb{\mathrm{2}}} \HOLSymConst{\HOLTokenImp{}} \HOLConst{SG} (\HOLTokenLambda{}\HOLBoundVar{t}. \HOLFreeVar{e\sb{\mathrm{1}}} \HOLBoundVar{t} \HOLSymConst{\ensuremath{+}} \HOLFreeVar{e\sb{\mathrm{2}}} \HOLBoundVar{t})
\HOLConst{SG} \HOLFreeVar{e\sb{\mathrm{1}}} \HOLSymConst{\HOLTokenConj{}} \HOLConst{SG} \HOLFreeVar{e\sb{\mathrm{2}}} \HOLSymConst{\HOLTokenImp{}} \HOLConst{SG} (\HOLTokenLambda{}\HOLBoundVar{t}. \HOLFreeVar{e\sb{\mathrm{1}}} \HOLBoundVar{t} \HOLSymConst{\ensuremath{\parallel}} \HOLFreeVar{e\sb{\mathrm{2}}} \HOLBoundVar{t})
\HOLConst{SG} \HOLFreeVar{e} \HOLSymConst{\HOLTokenImp{}} \HOLConst{SG} (\HOLTokenLambda{}\HOLBoundVar{t}. \HOLSymConst{\ensuremath{\nu}} \HOLFreeVar{L} (\HOLFreeVar{e} \HOLBoundVar{t}))
\HOLConst{SG} \HOLFreeVar{e} \HOLSymConst{\HOLTokenImp{}} \HOLConst{SG} (\HOLTokenLambda{}\HOLBoundVar{t}. \HOLConst{relab} (\HOLFreeVar{e} \HOLBoundVar{t}) \HOLFreeVar{rf})\hfill{[SG_rules]}
\end{alltt}

A context is \emph{sequential} (\texttt{SEQ}) if each of its \emph{subcontexts} with
a hole, apart from the hole itself, is in forms of prefixes or sums:
(c.f. Def.~\ref{def:guardness} and p.101,157 of \cite{Mil89} for
the informal definitions.)
 \begin{alltt}
\HOLConst{SEQ} (\HOLTokenLambda{}\HOLBoundVar{t}. \HOLBoundVar{t})
\HOLConst{SEQ} (\HOLTokenLambda{}\HOLBoundVar{t}. \HOLFreeVar{p})
\HOLConst{SEQ} \HOLFreeVar{e} \HOLSymConst{\HOLTokenImp{}} \HOLConst{SEQ} (\HOLTokenLambda{}\HOLBoundVar{t}. \HOLFreeVar{a}\HOLSymConst{..}\HOLFreeVar{e} \HOLBoundVar{t})
\HOLConst{SEQ} \HOLFreeVar{e\sb{\mathrm{1}}} \HOLSymConst{\HOLTokenConj{}} \HOLConst{SEQ} \HOLFreeVar{e\sb{\mathrm{2}}} \HOLSymConst{\HOLTokenImp{}} \HOLConst{SEQ} (\HOLTokenLambda{}\HOLBoundVar{t}. \HOLFreeVar{e\sb{\mathrm{1}}} \HOLBoundVar{t} \HOLSymConst{\ensuremath{+}} \HOLFreeVar{e\sb{\mathrm{2}}} \HOLBoundVar{t})\hfill{[SEQ_rules]}
\end{alltt}

In the same manner, \hl{we can also define variants of contexts (\texttt{GCONTEXT}) and weakly guarded
contexts (\texttt{WGS}) in which only guarded sums are allowed (i.e.~arbitrary sums are forbidden):}
\begin{alltt}
\HOLConst{GCONTEXT} (\HOLTokenLambda{}\HOLBoundVar{t}. \HOLBoundVar{t})
\HOLConst{GCONTEXT} (\HOLTokenLambda{}\HOLBoundVar{t}. \HOLFreeVar{p})
\HOLConst{GCONTEXT} \HOLFreeVar{e} \HOLSymConst{\HOLTokenImp{}} \HOLConst{GCONTEXT} (\HOLTokenLambda{}\HOLBoundVar{t}. \HOLFreeVar{a}\HOLSymConst{..}\HOLFreeVar{e} \HOLBoundVar{t})
\HOLConst{GCONTEXT} \HOLFreeVar{e\sb{\mathrm{1}}} \HOLSymConst{\HOLTokenConj{}} \HOLConst{GCONTEXT} \HOLFreeVar{e\sb{\mathrm{2}}} \HOLSymConst{\HOLTokenImp{}} \HOLConst{GCONTEXT} (\HOLTokenLambda{}\HOLBoundVar{t}. \HOLFreeVar{a\sb{\mathrm{1}}}\HOLSymConst{..}\HOLFreeVar{e\sb{\mathrm{1}}} \HOLBoundVar{t} \HOLSymConst{\ensuremath{+}} \HOLFreeVar{a\sb{\mathrm{2}}}\HOLSymConst{..}\HOLFreeVar{e\sb{\mathrm{2}}} \HOLBoundVar{t})
\HOLConst{GCONTEXT} \HOLFreeVar{e\sb{\mathrm{1}}} \HOLSymConst{\HOLTokenConj{}} \HOLConst{GCONTEXT} \HOLFreeVar{e\sb{\mathrm{2}}} \HOLSymConst{\HOLTokenImp{}} \HOLConst{GCONTEXT} (\HOLTokenLambda{}\HOLBoundVar{t}. \HOLFreeVar{e\sb{\mathrm{1}}} \HOLBoundVar{t} \HOLSymConst{\ensuremath{\parallel}} \HOLFreeVar{e\sb{\mathrm{2}}} \HOLBoundVar{t})
\HOLConst{GCONTEXT} \HOLFreeVar{e} \HOLSymConst{\HOLTokenImp{}} \HOLConst{GCONTEXT} (\HOLTokenLambda{}\HOLBoundVar{t}. \HOLSymConst{\ensuremath{\nu}} \HOLFreeVar{L} (\HOLFreeVar{e} \HOLBoundVar{t}))
\HOLConst{GCONTEXT} \HOLFreeVar{e} \HOLSymConst{\HOLTokenImp{}} \HOLConst{GCONTEXT} (\HOLTokenLambda{}\HOLBoundVar{t}. \HOLConst{relab} (\HOLFreeVar{e} \HOLBoundVar{t}) \HOLFreeVar{rf})\hfill{[GCONTEXT_rules]}
\end{alltt}
\begin{alltt}
\HOLConst{WGS} (\HOLTokenLambda{}\HOLBoundVar{t}. \HOLFreeVar{p})
\HOLConst{GCONTEXT} \HOLFreeVar{e} \HOLSymConst{\HOLTokenImp{}} \HOLConst{WGS} (\HOLTokenLambda{}\HOLBoundVar{t}. \HOLFreeVar{a}\HOLSymConst{..}\HOLFreeVar{e} \HOLBoundVar{t})
\HOLConst{GCONTEXT} \HOLFreeVar{e\sb{\mathrm{1}}} \HOLSymConst{\HOLTokenConj{}} \HOLConst{GCONTEXT} \HOLFreeVar{e\sb{\mathrm{2}}} \HOLSymConst{\HOLTokenImp{}} \HOLConst{WGS} (\HOLTokenLambda{}\HOLBoundVar{t}. \HOLFreeVar{a\sb{\mathrm{1}}}\HOLSymConst{..}\HOLFreeVar{e\sb{\mathrm{1}}} \HOLBoundVar{t} \HOLSymConst{\ensuremath{+}} \HOLFreeVar{a\sb{\mathrm{2}}}\HOLSymConst{..}\HOLFreeVar{e\sb{\mathrm{2}}} \HOLBoundVar{t})
\HOLConst{WGS} \HOLFreeVar{e\sb{\mathrm{1}}} \HOLSymConst{\HOLTokenConj{}} \HOLConst{WGS} \HOLFreeVar{e\sb{\mathrm{2}}} \HOLSymConst{\HOLTokenImp{}} \HOLConst{WGS} (\HOLTokenLambda{}\HOLBoundVar{t}. \HOLFreeVar{e\sb{\mathrm{1}}} \HOLBoundVar{t} \HOLSymConst{\ensuremath{\parallel}} \HOLFreeVar{e\sb{\mathrm{2}}} \HOLBoundVar{t})
\HOLConst{WGS} \HOLFreeVar{e} \HOLSymConst{\HOLTokenImp{}} \HOLConst{WGS} (\HOLTokenLambda{}\HOLBoundVar{t}. \HOLSymConst{\ensuremath{\nu}} \HOLFreeVar{L} (\HOLFreeVar{e} \HOLBoundVar{t}))
\HOLConst{WGS} \HOLFreeVar{e} \HOLSymConst{\HOLTokenImp{}} \HOLConst{WGS} (\HOLTokenLambda{}\HOLBoundVar{t}. \HOLConst{relab} (\HOLFreeVar{e} \HOLBoundVar{t}) \HOLFreeVar{rf})\hfill{[WGS_rules]}
\end{alltt}


A (pre)congruence is a relation on CCS processes defined on top of
\texttt{CONTEXT}. \hl{The only difference between congruence and
precongruence is that the former must be an equivalence (reflexive,
symmetric, transitive), while the latter can be just a preorder (reflexive, transitive)}:
\begin{alltt}
\HOLConst{congruence} \HOLFreeVar{R} \HOLSymConst{\HOLTokenEquiv{}}
\HOLConst{equivalence} \HOLFreeVar{R} \HOLSymConst{\HOLTokenConj{}}
\HOLSymConst{\HOLTokenForall{}}\HOLBoundVar{x} \HOLBoundVar{y} \HOLBoundVar{ctx}. \HOLConst{CONTEXT} \HOLBoundVar{ctx} \HOLSymConst{\HOLTokenImp{}} \HOLFreeVar{R} \HOLBoundVar{x} \HOLBoundVar{y} \HOLSymConst{\HOLTokenImp{}} \HOLFreeVar{R} (\HOLBoundVar{ctx} \HOLBoundVar{x}) (\HOLBoundVar{ctx} \HOLBoundVar{y})\hfill{[congruence_def]}
\hfill{[precongruence_def]}
\end{alltt}
\vspace{-4ex}
\begin{alltt}
\HOLConst{precongruence} \HOLFreeVar{R} \HOLSymConst{\HOLTokenEquiv{}}
\HOLConst{PreOrder} \HOLFreeVar{R} \HOLSymConst{\HOLTokenConj{}} \HOLSymConst{\HOLTokenForall{}}\HOLBoundVar{x} \HOLBoundVar{y} \HOLBoundVar{ctx}. \HOLConst{CONTEXT} \HOLBoundVar{ctx} \HOLSymConst{\HOLTokenImp{}} \HOLFreeVar{R} \HOLBoundVar{x} \HOLBoundVar{y} \HOLSymConst{\HOLTokenImp{}} \HOLFreeVar{R} (\HOLBoundVar{ctx} \HOLBoundVar{x}) (\HOLBoundVar{ctx} \HOLBoundVar{y})
\end{alltt}

\hl{For example, we can prove that, strong bisimilarity ($\sim$) and
rooted bisimilarity ($\approx^c$) are both congruence by above
definition: (the transitivity proof of rooted bisimilarity is actually not easy.)}
\begin{alltt}
\HOLTokenTurnstile{} \HOLConst{congruence} \HOLConst{STRONG_EQUIV}\hfill{[STRONG_EQUIV_congruence]}
\HOLTokenTurnstile{} \HOLConst{congruence} \HOLConst{OBS_CONGR}\hfill{[OBS_CONGR_congruence]}
\end{alltt}

\hl{Although weak bisimilarity ($\approx$) is \emph{not} congruence
  with respect to~\texttt{CONTEXT}, it is indeed ``congruence''
  with respect to~\texttt{GCONTEXT} (or if the CCS syntax were defined with
  only guarded sum operator \cite{sangiorgi2015equations}) as weak
  bisimilarity ($\approx$) is indeed preserved by 
  weakly-guarded sums.}



\subsection{Coarsest (pre)congruence contained in $\approx$ ($\succeq_{\mathrm{bis}}$)}
\label{s:coarsest}


As bisimilarity ($\approx$) is not congruence, for this reason rooted bisimilarity has been
introduced (Def.~\ref{d:rootedBisimilarity}).
In this subsection we discuss two proofs of an important result stating that
rooted bisimilarity is the coarsest congruence contained in
bisimilarity \cite{van2005characterisation,Gorrieri:2015jt,Mil89} (thus it
is the best one):
\begin{equation}
\label{eq:coarsest}
\forall p\ \ q.\ p\ \rapprox\ \! q\ \Longleftrightarrow\ ( \forall r.\ p\ +\
r\ \approx\ q\ +\ r )\enspace.
\end{equation}


Actually the coarsest congruence
contained in (weak) bisimilarity, namely the \emph{bisimilarity
  congruence} \cite{van2005characterisation}, can be constructed as
the \emph{composition closure} (\texttt{CC}) of (weak) bisimilarity:
\begin{alltt}
\HOLConst{WEAK_CONGR} \HOLSymConst{=} \HOLConst{CC} \HOLConst{WEAK_EQUIV}\hfill{[WEAK_CONGR]}
\HOLConst{CC} \HOLFreeVar{R} \HOLSymConst{=} (\HOLTokenLambda{}\HOLBoundVar{g} \HOLBoundVar{h}. \HOLSymConst{\HOLTokenForall{}}\HOLBoundVar{c}. \HOLConst{CONTEXT} \HOLBoundVar{c} \HOLSymConst{\HOLTokenImp{}} \HOLFreeVar{R} (\HOLBoundVar{c} \HOLBoundVar{g}) (\HOLBoundVar{c} \HOLBoundVar{h}))\hfill{[CC_def]}
\end{alltt}
\hl{Indeed, for any relation $R$ 
on CCS processes, the composition closure of $R$ is always finer (i.e.~smaller) than
$R$, no matter if $R$ is (pre)congruence or not}\footnote{\hl{But if $R$ is
  equivalence (or preorder), the composition closure of $R$ must be congruence
  (or precongruence). Also there is no need to put $R\ g\ h$ in the antecedent of
\texttt{CC\_def}, as this is anyhow obtained from the trivial context $(\lambda x.\,x)$.}}: (here $\subseteq_r$ stands for \emph{relational subset})
\begin{alltt}
\HOLTokenTurnstile{} \HOLSymConst{\HOLTokenForall{}}\HOLBoundVar{R}. \HOLConst{CC} \HOLBoundVar{R} \HOLSymConst{\HOLTokenRSubset{}} \HOLBoundVar{R}\hfill{[CC_is_finer]}
\end{alltt}
\hl{Furthermore, we prove that any (pre)congruence contained in $R$ (which
itself may not be) is contained in the composition closure of $R$
(hence the closure is the coarsest one):}
\begin{alltt}
\HOLTokenTurnstile{} \HOLSymConst{\HOLTokenForall{}}\HOLBoundVar{R} \HOLBoundVar{R\sp{\prime}}. \HOLConst{congruence} \HOLBoundVar{R\sp{\prime}} \HOLSymConst{\HOLTokenConj{}} \HOLBoundVar{R\sp{\prime}} \HOLSymConst{\HOLTokenRSubset{}} \HOLBoundVar{R} \HOLSymConst{\HOLTokenImp{}} \HOLBoundVar{R\sp{\prime}} \HOLSymConst{\HOLTokenRSubset{}} \HOLConst{CC} \HOLBoundVar{R}\hfill{[CC_is_coarsest]}
\HOLTokenTurnstile{} \HOLSymConst{\HOLTokenForall{}}\HOLBoundVar{R} \HOLBoundVar{R\sp{\prime}}. \HOLConst{precongruence} \HOLBoundVar{R\sp{\prime}} \HOLSymConst{\HOLTokenConj{}} \HOLBoundVar{R\sp{\prime}} \HOLSymConst{\HOLTokenRSubset{}} \HOLBoundVar{R} \HOLSymConst{\HOLTokenImp{}} \HOLBoundVar{R\sp{\prime}} \HOLSymConst{\HOLTokenRSubset{}} \HOLConst{CC} \HOLBoundVar{R}\hfill{[CC_is_coarsest']}
\end{alltt}

Given the central role of the  
 sum operator, we also consider the closure of bisimilarity under such
 operator, called \hl{\emph{equivalence compatible with sums}}
(\texttt{SUM_EQUIV}): 
\begin{alltt}
\HOLConst{SUM_EQUIV} \HOLSymConst{=} (\HOLTokenLambda{}\HOLBoundVar{p} \HOLBoundVar{q}. \HOLSymConst{\HOLTokenForall{}}\HOLBoundVar{r}. \HOLBoundVar{p} \HOLSymConst{\ensuremath{+}} \HOLBoundVar{r} \HOLSymConst{\HOLTokenWeakEQ} \HOLBoundVar{q} \HOLSymConst{\ensuremath{+}} \HOLBoundVar{r})\hfill{[SUM_EQUIV]}
\end{alltt}

\hl{Rooted bisimilarity $\rapprox$ (a congruence contained in
$\wb$), is now contained in \texttt{WEAK_CONGR},
which in turn is trivially contained in \texttt{SUM_EQUIV}}, as shown
in Fig.~\ref{fig:relationship}. Thus, to prove (\ref{eq:coarsest}),
the crux is to prove that \texttt{SUM_EQUIV} implies
rooted bisimilarity ($\rapprox$), making all three relations
($\rapprox$, \texttt{WEAK_CONGR} and \texttt{SUM_EQUIV}) equivalent:
\begin{equation}
\label{equa:pq}
\forall p\ \ q.\ ( \forall r.\ p\ +\ r \;\approx\; q\ +\ r ) \
\Rightarrow\ p\ \rapprox\ \! q\enspace.
\end{equation}

\begin{figure}[ht]
\begin{displaymath}
\xymatrix@R=3ex{
{\textrm{Weak bisimilarity } (\approx)} & {\textrm{Equiv.
    compatible with sums (\texttt{SUM\_EQUIV})}} \ar@/^3ex/[ldd]^{\subseteq}\\
{\textrm{Bisimilarity congruence (\texttt{WEAK\_CONGR})}}
\ar[u]^{\subseteq} \ar[ru]^{\subseteq} \\
{\textrm{Rooted bisimilarity } (\rapprox)} \ar[u]^{\subseteq}
}
\end{displaymath}
\vspace{-2ex}
\caption{\hl{Relationship between the equivalences mentioned}}
\label{fig:relationship}
\end{figure}
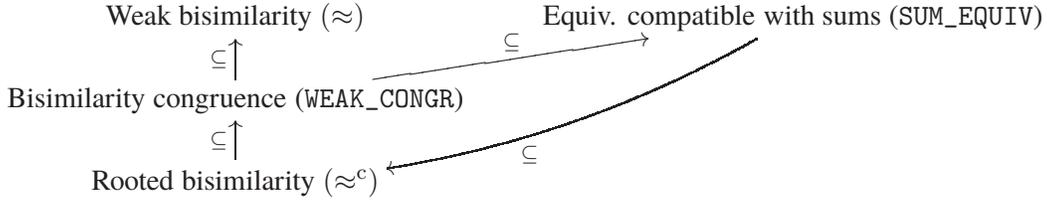

The standard argument \cite{Mil89} requires that $p$
and $q$ do not use up all available labels (i.e.~visible actions).
Formalising such an argument requires however 
a detailed treatment on free and bound names of CCS
processes (with the restriction operator being a binder), not done yet.
However, the proof of (\ref{equa:pq}) can be carried out 
just assuming that all immediate \emph{weak} derivatives of
 $p$ and $q$ do not use up all \hl{available labels}.
We have formalised this property and
 called it the \emph{free action} property:
\begin{alltt}
\HOLConst{free_action} \HOLFreeVar{p} \HOLSymConst{\HOLTokenEquiv{}} \HOLSymConst{\HOLTokenExists{}}\HOLBoundVar{a}. \HOLSymConst{\HOLTokenForall{}}\HOLBoundVar{p\sp{\prime}}. \HOLSymConst{\HOLTokenNeg{}}(\HOLFreeVar{p} \HOLTokenWeakTransBegin\HOLConst{label} \HOLBoundVar{a}\HOLTokenWeakTransEnd \HOLBoundVar{p\sp{\prime}})\hfill{[free_action_def]}
\end{alltt}
With this property, the actual formalisation of (\ref{equa:pq}) says:
\vspace{-2ex}
\begin{alltt}
\hfill{[COARSEST_CONGR_RL]}
\HOLTokenTurnstile{} \HOLConst{free_action} \HOLFreeVar{p} \HOLSymConst{\HOLTokenConj{}} \HOLConst{free_action} \HOLFreeVar{q} \HOLSymConst{\HOLTokenImp{}} (\HOLSymConst{\HOLTokenForall{}}\HOLBoundVar{r}. \HOLFreeVar{p} \HOLSymConst{\ensuremath{+}} \HOLBoundVar{r} \HOLSymConst{\HOLTokenWeakEQ} \HOLFreeVar{q} \HOLSymConst{\ensuremath{+}} \HOLBoundVar{r}) \HOLSymConst{\HOLTokenImp{}} \HOLFreeVar{p} \HOLSymConst{\HOLTokenObsCongr} \HOLFreeVar{q}
\end{alltt}

With an almost identical proof, rooted contraction
($\rcontr$) is also the coarsest
precongruence contained in bisimilarity contraction ($\mcontrBIS$)
(the other direction of (\ref{eq:coarsest}) is trivial):
\vspace{-2ex}
\begin{alltt}
\hfill{[COARSEST_PRECONGR_RL]}
\HOLTokenTurnstile{} \HOLConst{free_action} \HOLFreeVar{p} \HOLSymConst{\HOLTokenConj{}} \HOLConst{free_action} \HOLFreeVar{q} \HOLSymConst{\HOLTokenImp{}} (\HOLSymConst{\HOLTokenForall{}}\HOLBoundVar{r}. \HOLFreeVar{p} \HOLSymConst{\ensuremath{+}} \HOLBoundVar{r} \HOLSymConst{\HOLTokenContracts{}} \HOLFreeVar{q} \HOLSymConst{\ensuremath{+}} \HOLBoundVar{r}) \HOLSymConst{\HOLTokenImp{}} \HOLFreeVar{p} \HOLSymConst{\HOLTokenObsContracts} \HOLFreeVar{q}
\end{alltt}

The formal proofs of above two results precisely follow Milner
\cite{Mil89}. \hl{If only $p$ (or $q$) has free actions while the other uses up all available
labels, the classic assumption $\mathrm{fn}(p) \cup
\mathrm{fn}(q) \neq \mathscr{L}$ (here $\mathrm{fn}$ stands for \emph{free
  names}) does not hold, and the proof cannot be completed.} Our
assumption is a bit \emph{weaker}
in the sense that, $p$ and $q$ do not really need to have the
\emph{same} free action (also, $a$ and $\outC a$ are
\emph{different} actions).

\hl{There exists a different, more complex proof of (\ref{eq:coarsest}),
given by van Glabbeek} \cite{van2005characterisation}, which does not
require any additional assumption.
The core lemma says, for
any two processes $p$ and $q$, if there exists a \emph{stable} (i.e.~$\tau$-free)
 process $k$ which is not bisimilar with any derivative of $p$ and
 $q$, then \texttt{SUM_EQUIV} indeed implies rooted bisimilarity ($\rapprox$):
\begin{alltt}
\HOLTokenTurnstile{} (\HOLSymConst{\HOLTokenExists{}}\HOLBoundVar{k}.
        \HOLConst{STABLE} \HOLBoundVar{k} \HOLSymConst{\HOLTokenConj{}} (\HOLSymConst{\HOLTokenForall{}}\HOLBoundVar{p\sp{\prime}} \HOLBoundVar{u}. \HOLFreeVar{p} \HOLTokenWeakTransBegin\HOLBoundVar{u}\HOLTokenWeakTransEnd \HOLBoundVar{p\sp{\prime}} \HOLSymConst{\HOLTokenImp{}} \HOLSymConst{\HOLTokenNeg{}}(\HOLBoundVar{p\sp{\prime}} \HOLSymConst{\HOLTokenWeakEQ} \HOLBoundVar{k})) \HOLSymConst{\HOLTokenConj{}}
        \HOLSymConst{\HOLTokenForall{}}\HOLBoundVar{q\sp{\prime}} \HOLBoundVar{u}. \HOLFreeVar{q} \HOLTokenWeakTransBegin\HOLBoundVar{u}\HOLTokenWeakTransEnd \HOLBoundVar{q\sp{\prime}} \HOLSymConst{\HOLTokenImp{}} \HOLSymConst{\HOLTokenNeg{}}(\HOLBoundVar{q\sp{\prime}} \HOLSymConst{\HOLTokenWeakEQ} \HOLBoundVar{k})) \HOLSymConst{\HOLTokenImp{}}
   (\HOLSymConst{\HOLTokenForall{}}\HOLBoundVar{r}. \HOLFreeVar{p} \HOLSymConst{\ensuremath{+}} \HOLBoundVar{r} \HOLSymConst{\HOLTokenWeakEQ} \HOLFreeVar{q} \HOLSymConst{\ensuremath{+}} \HOLBoundVar{r}) \HOLSymConst{\HOLTokenImp{}}
   \HOLFreeVar{p} \HOLSymConst{\HOLTokenObsCongr} \HOLFreeVar{q}\hfill{[PROP3_COMMON]}
\end{alltt}
\begin{alltt}
\HOLConst{STABLE} \HOLFreeVar{p} \HOLSymConst{\HOLTokenEquiv{}} \HOLSymConst{\HOLTokenForall{}}\HOLBoundVar{u} \HOLBoundVar{p\sp{\prime}}. \HOLFreeVar{p} \HOLTokenTransBegin\HOLBoundVar{u}\HOLTokenTransEnd \HOLBoundVar{p\sp{\prime}} \HOLSymConst{\HOLTokenImp{}} \HOLBoundVar{u} \HOLSymConst{\HOLTokenNotEqual{}} \HOLSymConst{\ensuremath{\tau}}\hfill{[STABLE]}
\end{alltt}
\hl{To actually get this process $k$, the proof relies on arbitrary infinite sum of 
processes and uses transfinite induction to obtain
an arbitrary large sequence of processes (firstly introduced by Jan
Willem Klop \cite{van2005characterisation})
 that are all pairwise non-bisimilar.} 
We have partially formalised
this proof, because the typed logic
implemented in various HOL systems (including Isabelle/HOL) is not
strong enough to define a type for all possible ordinal values
\cite{norrish2013ordinals}, thus
we have replaced transfinite induction with plain induction. As a
consequence, the final
result is about a restricted class of processes (which we have taken
to be the finite-state processes). This proof uses extensively HOL's
\texttt{pred_set} theory \cite{melham1992hol} and has an interesting mix
of CCS and pure mathematics in it. (c.f. \cite{Tian:2017wrba} for
more details.)


\subsection{Unique solution of contractions}

A delicate point in the formalisation of the results about unique solution of
contractions are the proof of Lemma~\ref{l:ruptocon} and lemmas alike;
in particular, there is
 an induction on the length of weak transitions. 
For this, rather than 
 introducing a refined form of weak transition relation
enriched with its length, 
we found it more elegant  to  work with traces
(a motivation for this is to set the ground for extensions of this
formalisation work to trace equivalence in place of bisimilarity).



A trace is represented by the initial and final processes, plus
a list of actions  so performed.
For this, we first 
 define \hl{the concept of label-accumulated reflexive transitive closure
 (\texttt{LRTC})}.
Given a labeled transition relation \texttt{R} on CCS, \texttt{LRTC R} is
a label-accumulated relation representing the trace of transitions:
\begin{alltt}
\HOLConst{LRTC} \HOLFreeVar{R} \HOLFreeVar{a} \HOLFreeVar{l} \HOLFreeVar{b} \HOLSymConst{\HOLTokenEquiv{}}
\HOLSymConst{\HOLTokenForall{}}\HOLBoundVar{P}.
    (\HOLSymConst{\HOLTokenForall{}}\HOLBoundVar{x}. \HOLBoundVar{P} \HOLBoundVar{x} [] \HOLBoundVar{x}) \HOLSymConst{\HOLTokenConj{}}
    (\HOLSymConst{\HOLTokenForall{}}\HOLBoundVar{x} \HOLBoundVar{h} \HOLBoundVar{y} \HOLBoundVar{t} \HOLBoundVar{z}. \HOLFreeVar{R} \HOLBoundVar{x} \HOLBoundVar{h} \HOLBoundVar{y} \HOLSymConst{\HOLTokenConj{}} \HOLBoundVar{P} \HOLBoundVar{y} \HOLBoundVar{t} \HOLBoundVar{z} \HOLSymConst{\HOLTokenImp{}} \HOLBoundVar{P} \HOLBoundVar{x} (\HOLBoundVar{h}\HOLSymConst{::}\HOLBoundVar{t}) \HOLBoundVar{z}) \HOLSymConst{\HOLTokenImp{}}
    \HOLBoundVar{P} \HOLFreeVar{a} \HOLFreeVar{l} \HOLFreeVar{b}\hfill{[LRTC_DEF]}
\end{alltt}
\hl{The trace relation for CCS can be then obtained
 by calling \texttt{LRTC} on the (strong, or single-step) labeled transition
 relation \texttt{TRANS} ($\overset{\mu}{\rightarrow}$) defined by SOS rules}:
\begin{alltt}
\HOLConst{TRACE} \HOLSymConst{=} \HOLConst{LRTC} \HOLConst{TRANS}\hfill{[TRACE_def]}
\end{alltt}

\hl{If the list of actions is empty, that means that there is no transition and hence,}
if there is at most one visible action (i.e., a label) in the list of actions,
then the trace is also a weak transition. Here
we have to distinguish between two cases: no label and unique label (in
the list of actions). The definition of ``no
label'' in an action list is easy (here \texttt{MEM} tests if a given element is a member of a list):
\begin{alltt}
\HOLConst{NO_LABEL} \HOLFreeVar{L} \HOLSymConst{\HOLTokenEquiv{}} \HOLSymConst{\HOLTokenNeg{}}\HOLSymConst{\HOLTokenExists{}}\HOLBoundVar{l}. \HOLConst{MEM} (\HOLConst{label} \HOLBoundVar{l}) \HOLFreeVar{L}\hfill{[NO_LABEL_def]}
\end{alltt}

The definition of ``unique label'' \hl{can be done in many ways, the
following definition (a suggestion from Robert Beers)
avoids any counting or filtering in the list.}
It says that a label is unique in a list of actions if and only if there is no
label in the rest of list:
\begin{alltt}
\HOLConst{UNIQUE_LABEL} \HOLFreeVar{u} \HOLFreeVar{L} \HOLSymConst{\HOLTokenEquiv{}}
\HOLSymConst{\HOLTokenExists{}}\HOLBoundVar{L\sb{\mathrm{1}}} \HOLBoundVar{L\sb{\mathrm{2}}}. (\HOLBoundVar{L\sb{\mathrm{1}}} \HOLSymConst{\HOLTokenDoublePlus} [\HOLFreeVar{u}] \HOLSymConst{\HOLTokenDoublePlus} \HOLBoundVar{L\sb{\mathrm{2}}} \HOLSymConst{=} \HOLFreeVar{L}) \HOLSymConst{\HOLTokenConj{}} \HOLConst{NO_LABEL} \HOLBoundVar{L\sb{\mathrm{1}}} \HOLSymConst{\HOLTokenConj{}} \HOLConst{NO_LABEL} \HOLBoundVar{L\sb{\mathrm{2}}}\hfill{[UNIQUE_LABEL_def]}
\end{alltt}

The final relationship between traces and weak transitions is stated
and proved in the following theorem
(where the  variable $acts$ stands for
a list of actions); 
it says, a weak transition $P\overset{u}{\Rightarrow}P'$ is also a
trace $P\overset{acts}{\longrightarrow}P'$ with a
 non-empty action list $acts$, in which either there is no label (for $u = \tau$), or 
$u$ is the unique label (for $u \neq \tau$):
\begin{alltt}
\HOLTokenTurnstile{} \HOLFreeVar{P} \HOLTokenWeakTransBegin\HOLFreeVar{u}\HOLTokenWeakTransEnd \HOLFreeVar{P\sp{\prime}} \HOLSymConst{\HOLTokenEquiv{}}
   \HOLSymConst{\HOLTokenExists{}}\HOLBoundVar{acts}.
       \HOLConst{TRACE} \HOLFreeVar{P} \HOLBoundVar{acts} \HOLFreeVar{P\sp{\prime}} \HOLSymConst{\HOLTokenConj{}} \HOLSymConst{\HOLTokenNeg{}}\HOLConst{NULL} \HOLBoundVar{acts} \HOLSymConst{\HOLTokenConj{}}
       \HOLKeyword{if} \HOLFreeVar{u} \HOLSymConst{=} \HOLSymConst{\ensuremath{\tau}} \HOLKeyword{then} \HOLConst{NO_LABEL} \HOLBoundVar{acts} \HOLKeyword{else} \HOLConst{UNIQUE_LABEL} \HOLFreeVar{u} \HOLBoundVar{acts}\hfill{[WEAK_TRANS_AND_TRACE]}
\end{alltt}

Now the formalised version of Lemma~\ref{l:uptocon}:
\hfill{\texttt{[UNIQUE_SOLUTION_OF_CONTRACTIONS_LEMMA]}\vspace{-1em}
\begin{alltt}
\begin{small}
\HOLTokenTurnstile{} (\HOLSymConst{\HOLTokenExists{}}\HOLBoundVar{E}. \HOLConst{WGS} \HOLBoundVar{E} \HOLSymConst{\HOLTokenConj{}} \HOLFreeVar{P} \HOLSymConst{\HOLTokenContracts{}} \HOLBoundVar{E} \HOLFreeVar{P} \HOLSymConst{\HOLTokenConj{}} \HOLFreeVar{Q} \HOLSymConst{\HOLTokenContracts{}} \HOLBoundVar{E} \HOLFreeVar{Q}) \HOLSymConst{\HOLTokenImp{}}
   \HOLSymConst{\HOLTokenForall{}}\HOLBoundVar{C}.
       \HOLConst{GCONTEXT} \HOLBoundVar{C} \HOLSymConst{\HOLTokenImp{}}
       (\HOLSymConst{\HOLTokenForall{}}\HOLBoundVar{l} \HOLBoundVar{R}.
            \HOLBoundVar{C} \HOLFreeVar{P} \HOLTokenWeakTransBegin\HOLConst{label} \HOLBoundVar{l}\HOLTokenWeakTransEnd \HOLBoundVar{R} \HOLSymConst{\HOLTokenImp{}}
            \HOLSymConst{\HOLTokenExists{}}\HOLBoundVar{C\sp{\prime}}.
                \HOLConst{GCONTEXT} \HOLBoundVar{C\sp{\prime}} \HOLSymConst{\HOLTokenConj{}} \HOLBoundVar{R} \HOLSymConst{\HOLTokenContracts{}} \HOLBoundVar{C\sp{\prime}} \HOLFreeVar{P} \HOLSymConst{\HOLTokenConj{}}
                (\HOLConst{WEAK_EQUIV} \HOLSymConst{\HOLTokenRCompose{}} (\HOLTokenLambda{}\HOLBoundVar{x} \HOLBoundVar{y}. \HOLBoundVar{x} \HOLTokenWeakTransBegin\HOLConst{label} \HOLBoundVar{l}\HOLTokenWeakTransEnd \HOLBoundVar{y})) (\HOLBoundVar{C} \HOLFreeVar{Q})
                  (\HOLBoundVar{C\sp{\prime}} \HOLFreeVar{Q})) \HOLSymConst{\HOLTokenConj{}}
       \HOLSymConst{\HOLTokenForall{}}\HOLBoundVar{R}.
           \HOLBoundVar{C} \HOLFreeVar{P} \HOLTokenWeakTransBegin\HOLSymConst{\ensuremath{\tau}}\HOLTokenWeakTransEnd \HOLBoundVar{R} \HOLSymConst{\HOLTokenImp{}}
           \HOLSymConst{\HOLTokenExists{}}\HOLBoundVar{C\sp{\prime}}.
               \HOLConst{GCONTEXT} \HOLBoundVar{C\sp{\prime}} \HOLSymConst{\HOLTokenConj{}} \HOLBoundVar{R} \HOLSymConst{\HOLTokenContracts{}} \HOLBoundVar{C\sp{\prime}} \HOLFreeVar{P} \HOLSymConst{\HOLTokenConj{}}
               (\HOLConst{WEAK_EQUIV} \HOLSymConst{\HOLTokenRCompose{}} \HOLConst{EPS}) (\HOLBoundVar{C} \HOLFreeVar{Q}) (\HOLBoundVar{C\sp{\prime}} \HOLFreeVar{Q})
\end{small}
\end{alltt}
\vspace{-1em}
Traces are actually used in the proof of above lemma via 
the following ``unfolding lemma'':\vspace{-1em}
\begin{alltt}
\begin{small}
\HOLTokenTurnstile{} \HOLConst{GCONTEXT} \HOLFreeVar{C} \HOLSymConst{\HOLTokenConj{}} \HOLConst{WGS} \HOLFreeVar{E} \HOLSymConst{\HOLTokenConj{}} \HOLConst{TRACE} ((\HOLFreeVar{C} \HOLSymConst{\HOLTokenCompose} \HOLConst{FUNPOW} \HOLFreeVar{E} \HOLFreeVar{n}) \HOLFreeVar{P}) \HOLFreeVar{xs} \HOLFreeVar{P\sp{\prime}} \HOLSymConst{\HOLTokenConj{}}
   \HOLConst{LENGTH} \HOLFreeVar{xs} \HOLSymConst{\HOLTokenLeq{}} \HOLFreeVar{n} \HOLSymConst{\HOLTokenImp{}}
   \HOLSymConst{\HOLTokenExists{}}\HOLBoundVar{C\sp{\prime}}.
       \HOLConst{GCONTEXT} \HOLBoundVar{C\sp{\prime}} \HOLSymConst{\HOLTokenConj{}} (\HOLFreeVar{P\sp{\prime}} \HOLSymConst{=} \HOLBoundVar{C\sp{\prime}} \HOLFreeVar{P}) \HOLSymConst{\HOLTokenConj{}}
       \HOLSymConst{\HOLTokenForall{}}\HOLBoundVar{Q}. \HOLConst{TRACE} ((\HOLFreeVar{C} \HOLSymConst{\HOLTokenCompose} \HOLConst{FUNPOW} \HOLFreeVar{E} \HOLFreeVar{n}) \HOLBoundVar{Q}) \HOLFreeVar{xs} (\HOLBoundVar{C\sp{\prime}} \HOLBoundVar{Q})\hfill{[unfolding_lemma4]}
\end{small}
\end{alltt}
\vspace{-1em}
It roughly says, for any context $C$ and weakly-guarded context
$E$, if $C [\, E^n[P]\,] \overset{xs}{\Longrightarrow} P'$ and the length
of actions $xs \leqslant n$, then $P$ has the form of $C'[P]$ (meaning
that $P$ is not touched during the transitions). Traces are used for
reasoning about the \hl{number} of intermediate actions in weak
transitions. For instance, from Def.~\ref{d:BisCon}, \hl{it is easy
to see that, a weak transition either becomes shorter
or remains the same when moving between $\mcontrBIS$-related processes}.
\hl{This property is essential} in the proof of
Lemma~\ref{l:uptocon}. We show only one such lemma, for the case of
non-$\tau$ weak transitions passing into $\mcontrBIS$ (from left to right):
\begin{alltt}
\HOLTokenTurnstile{} \HOLFreeVar{P} \HOLSymConst{\HOLTokenContracts{}} \HOLFreeVar{Q} \HOLSymConst{\HOLTokenImp{}}
   \HOLSymConst{\HOLTokenForall{}}\HOLBoundVar{xs} \HOLBoundVar{l} \HOLBoundVar{P\sp{\prime}}.
       \HOLConst{TRACE} \HOLFreeVar{P} \HOLBoundVar{xs} \HOLBoundVar{P\sp{\prime}} \HOLSymConst{\HOLTokenConj{}} \HOLConst{UNIQUE_LABEL} (\HOLConst{label} \HOLBoundVar{l}) \HOLBoundVar{xs} \HOLSymConst{\HOLTokenImp{}}
       \HOLSymConst{\HOLTokenExists{}}\HOLBoundVar{xs\sp{\prime}} \HOLBoundVar{Q\sp{\prime}}.
           \HOLConst{TRACE} \HOLFreeVar{Q} \HOLBoundVar{xs\sp{\prime}} \HOLBoundVar{Q\sp{\prime}} \HOLSymConst{\HOLTokenConj{}} \HOLFreeVar{P} \HOLSymConst{\HOLTokenContracts{}} \HOLFreeVar{Q} \HOLSymConst{\HOLTokenConj{}} \HOLConst{LENGTH} \HOLBoundVar{xs\sp{\prime}} \HOLSymConst{\HOLTokenLeq{}} \HOLConst{LENGTH} \HOLBoundVar{xs} \HOLSymConst{\HOLTokenConj{}}
           \HOLConst{UNIQUE_LABEL} (\HOLConst{label} \HOLBoundVar{l}) \HOLBoundVar{xs\sp{\prime}}\hfill{[contracts_AND_TRACE_label]}
\end{alltt}

\hl{With all above lemmas, we can thus finally prove Theorem~\ref{t:contraBisimulationU}:}
\begin{alltt}
\HOLTokenTurnstile{} \HOLConst{WGS} \HOLFreeVar{E} \HOLSymConst{\HOLTokenImp{}} \HOLSymConst{\HOLTokenForall{}}\HOLBoundVar{P} \HOLBoundVar{Q}. \HOLBoundVar{P} \HOLSymConst{\HOLTokenContracts{}} \HOLFreeVar{E} \HOLBoundVar{P} \HOLSymConst{\HOLTokenConj{}} \HOLBoundVar{Q} \HOLSymConst{\HOLTokenContracts{}} \HOLFreeVar{E} \HOLBoundVar{Q} \HOLSymConst{\HOLTokenImp{}} \HOLBoundVar{P} \HOLSymConst{\HOLTokenWeakEQ} \HOLBoundVar{Q}
\hfill{[UNIQUE_SOLUTION_OF_CONTRACTIONS]}
\end{alltt}
\vspace{-2ex}

\subsection{Unique solution of rooted contractions}

The formal proof of ``unique solution of rooted contractions theorem''
(Theorem~\ref{t:rcontraBisimulationU}) has the
same initial proof steps as Theorem~\ref{t:contraBisimulationU}; 
it then requires a
few more steps to handle  rooted bisimilarity in the conclusion. 
Overall  the
two proofs are very similar, mostly because the only property we need
from (rooted) contraction is its precongruence. 
 Below is the formally verified version of
Theorem~\ref{t:rcontraBisimulationU}
(having proved
the precongruence of rooted contraction, 
we can use  weakly-guarded expressions,   without constraints on  sums;
that is, \texttt{WG} in place of \texttt{WGS}):
\vspace{-2ex}
\begin{alltt}
\hfill{[UNIQUE_SOLUTION_OF_ROOTED_CONTRACTIONS]}
\HOLTokenTurnstile{} \HOLConst{WG} \HOLFreeVar{E} \HOLSymConst{\HOLTokenImp{}} \HOLSymConst{\HOLTokenForall{}}\HOLBoundVar{P} \HOLBoundVar{Q}. \HOLBoundVar{P} \HOLSymConst{\HOLTokenObsContracts} \HOLFreeVar{E} \HOLBoundVar{P} \HOLSymConst{\HOLTokenConj{}} \HOLBoundVar{Q} \HOLSymConst{\HOLTokenObsContracts} \HOLFreeVar{E} \HOLBoundVar{Q} \HOLSymConst{\HOLTokenImp{}} \HOLBoundVar{P} \HOLSymConst{\HOLTokenObsCongr} \HOLBoundVar{Q}
\end{alltt}

Having removed the  constraints on sums, the result is
 similar to Milner's original `unique solution of
equations' theorem for \emph{strong} bisimilarity ($\sim$)~--- 
the same weakly guarded context (\texttt{WG}) is required:
\begin{alltt}
\HOLTokenTurnstile{} \HOLConst{WG} \HOLFreeVar{E} \HOLSymConst{\HOLTokenImp{}} \HOLSymConst{\HOLTokenForall{}}\HOLBoundVar{P} \HOLBoundVar{Q}. \HOLBoundVar{P} \HOLSymConst{\HOLTokenStrongEQ} \HOLFreeVar{E} \HOLBoundVar{P} \HOLSymConst{\HOLTokenConj{}} \HOLBoundVar{Q} \HOLSymConst{\HOLTokenStrongEQ} \HOLFreeVar{E} \HOLBoundVar{Q} \HOLSymConst{\HOLTokenImp{}} \HOLBoundVar{P} \HOLSymConst{\HOLTokenStrongEQ} \HOLBoundVar{Q}\hfill{[STRONG_UNIQUE_SOLUTION]}
\end{alltt}

In contrast, Milner's ``unique solution of
equations'' theorem for rooted bisimilarity ($\rapprox$)
has more severe constraints (must be both strongly guarded and sequential):
\vspace{-2ex}
\begin{alltt}
\hfill{[OBS_UNIQUE_SOLUTION]}
\HOLTokenTurnstile{} \HOLConst{SG} \HOLFreeVar{E} \HOLSymConst{\HOLTokenConj{}} \HOLConst{SEQ} \HOLFreeVar{E} \HOLSymConst{\HOLTokenImp{}} \HOLSymConst{\HOLTokenForall{}}\HOLBoundVar{P} \HOLBoundVar{Q}. \HOLBoundVar{P} \HOLSymConst{\HOLTokenObsCongr} \HOLFreeVar{E} \HOLBoundVar{P} \HOLSymConst{\HOLTokenConj{}} \HOLBoundVar{Q} \HOLSymConst{\HOLTokenObsCongr} \HOLFreeVar{E} \HOLBoundVar{Q} \HOLSymConst{\HOLTokenImp{}} \HOLBoundVar{P} \HOLSymConst{\HOLTokenObsCongr} \HOLBoundVar{Q}
\end{alltt}

\section{Related work on formalisation}
\label{s:rel}

\hl{Monica Nesi did the first CCS formalisations for both pure and
value-passing CCS} \cite{Nesi:1992ve,Nesi:2017wo} using early versions of the HOL
theorem prover.\footnote{Part of this work can now be found at
  \url{https://github.com/binghe/HOL-CCS/tree/master/CCS-Nesi}.}
Her main focus was on implementing decision procedures (as a ML
program, e.g.~\cite{cleaveland1993concurrency}) for
automatically proving bisimilarities of CCS
processes. 
\hl{Her work is
  the working basis of ours. However, the differences are substantial, especially in the way of defining
bisimilarities. We greatly benefited from features and standard
libraries in recent versions of HOL4, and our formalisation has
covered a much larger spectrum of the (pure) CCS theory.}

Bengtson, Parrow and Weber did a substantial formalisation work
on CCS, $\pi$-calculi
and $\psi$-calculi 
using Isabelle/HOL and its \texttt{nominal} logic, with main focus on the handling of
name binders \cite{bengtson2010formalising,bengtson2007completeness,parrow2009formalising}.
%
Other formalisations in this area include the early work of T. F. Melham
\cite{melham1994mechanized} and O.A. Mohamed
\cite{mohamed1995mechanizing} in HOL, Compton
\cite{compton2005embedding} in Isabelle/HOL,
Solange\footnote{\url{https://github.com/coq-contribs/ccs}} in Coq
and Chaudhuri et al.\;\cite{chaudhuri2015lightweight} in Abella, the latter
focuses on `bisimulation up-to' techniques 
for CCS and $\pi$-calculus.
Damien Pous \cite{pous2007new} also formalised up-to techniques and some CCS examples in
Coq.
Formalisations less related to ours
include Kahsai and Miculan \cite{kahsai2008implementing} for the spi
calculus in Isabelle/HOL, and Hirschkoff \cite{hirschkoff1997full} for the $\pi$-calculus in Coq.


\section{Conclusions and future work}
\label{s:concl}

In this paper, we have highlighted a formalisation of the theory of CCS in the 
HOL4 theorem prover (for lack of space we have not discussed 
the formalisation of some basic algebraic theory, of the basic
properties of the expansion preorder,   and of a few
 versions of `bisimulation up to'
techniques). 
The formalisation has focused on the theory of
unique solution of equations and contractions. 
It has also allowed us to further develop the theory,
notably the basic properties of rooted contraction, and the unique
solution theorem for it with respect to rooted bisimilarity. 
The formalisation brings up and exploits similarities between results
and proofs for different equivalences and preorders. 
We think that the statements in the formalisation are easy to read and
understand, as they are very close to the original statements found in
standard CCS textbooks \cite{Gorrieri:2015jt,Mil89}.

For the future work, it would be worth extending
to multi-variable equations/contractions. \hl{A key aspect could be using unguarded constants as free variables
(\texttt{FV}) and defining guardedness directly on expressions of type CCS (instead of
CCS $\rightarrow$ CCS), then linking to contexts. For instance, an expression is weakly-guarded when each
of its free variables, replaced by a hole, results in a weakly-guarded context:}
\begin{alltt}
\HOLTokenTurnstile{} \HOLConst{weakly_guarded1} \HOLFreeVar{E} \HOLSymConst{\HOLTokenEquiv{}}
   \HOLSymConst{\HOLTokenForall{}}\HOLBoundVar{X}. \HOLBoundVar{X} \HOLSymConst{\HOLTokenIn{}} \HOLConst{FV} \HOLFreeVar{E} \HOLSymConst{\HOLTokenImp{}} \HOLSymConst{\HOLTokenForall{}}\HOLBoundVar{e}. \HOLConst{CONTEXT} \HOLBoundVar{e} \HOLSymConst{\HOLTokenConj{}} (\HOLBoundVar{e} (\HOLConst{var} \HOLBoundVar{X}) \HOLSymConst{=} \HOLFreeVar{E}) \HOLSymConst{\HOLTokenImp{}} \HOLConst{WG} \HOLBoundVar{e}
\end{alltt}

Formalising other equivalences and preorders could also be considered,
notably the trace equivalences, as well as more refined process
calculi such as value-passing CCS.
%
On another research line, one could examine the formalisation of a different
approach \cite{DurierHS17} to unique
solutions, in which the use of contraction is
replaced by semantic conditions on process transitions such as
divergence. 



\paragraph{Acknowledgements}

We have benefitted from suggestions and comments 
from several people from the HOL
community, including (in alphabet order) Robert Beers, Jeremy Dawson,
Ramana Kumar,
Michael Norrish, 
Konrad Slind, and
Thomas T\"{u}rk.
The second half of this
paper was written in memory of Michael J.~C.~Gordon, the creator of HOL theorem prover.

\newpage

\bibliographystyle{eptcs}
\bibliography{generic}

\begin{thebibliography}{10}
\providecommand{\bibitemdeclare}[2]{}
\providecommand{\surnamestart}{}
\providecommand{\surnameend}{}
\providecommand{\urlprefix}{Available at }
\providecommand{\url}[1]{\texttt{#1}}
\providecommand{\href}[2]{\texttt{#2}}
\providecommand{\urlalt}[2]{\href{#1}{#2}}
\providecommand{\doi}[1]{doi:\urlalt{http://dx.doi.org/#1}{#1}}
\providecommand{\bibinfo}[2]{#2}

\bibitemdeclare{article}{arun1992efficiency}
\bibitem{arun1992efficiency}
\bibinfo{author}{S.~\surnamestart Arun-Kumar\surnameend} \&
  \bibinfo{author}{Matthew \surnamestart Hennessy\surnameend}
  (\bibinfo{year}{1992}): \emph{\bibinfo{title}{An efficiency preorder for
  processes}}.
\newblock {\sl \bibinfo{journal}{Acta Informatica}}
  \bibinfo{volume}{29}(\bibinfo{number}{8}), pp. \bibinfo{pages}{737--760},
  \doi{10.1007/BF01191894}.

\bibitemdeclare{book}{BaeBOOK}
\bibitem{BaeBOOK}
\bibinfo{author}{Jos C.~M. \surnamestart Baeten\surnameend},
  \bibinfo{author}{Twan \surnamestart Basten\surnameend} \&
  \bibinfo{author}{Michel~A. \surnamestart Reniers\surnameend}
  (\bibinfo{year}{2010}): \emph{\bibinfo{title}{Process Algebra: Equational
  Theories of Communicating Processes}}.
\newblock \bibinfo{publisher}{Cambridge University Press},
  \doi{10.1017/CBO9781139195003}.

\bibitemdeclare{phdthesis}{bengtson2010formalising}
\bibitem{bengtson2010formalising}
\bibinfo{author}{Jesper \surnamestart Bengtson\surnameend}
  (\bibinfo{year}{2010}): \emph{\bibinfo{title}{Formalising process calculi}}.
\newblock Ph.D. thesis, \bibinfo{school}{Acta Universitatis Upsaliensis}.

\bibitemdeclare{article}{bengtson2007completeness}
\bibitem{bengtson2007completeness}
\bibinfo{author}{Jesper \surnamestart Bengtson\surnameend} \&
  \bibinfo{author}{Joachim \surnamestart Parrow\surnameend}
  (\bibinfo{year}{2007}): \emph{\bibinfo{title}{A completeness proof for
  bisimulation in the pi-calculus using isabelle}}.
\newblock {\sl \bibinfo{journal}{Electronic Notes in Theoretical Computer
  Science}} \bibinfo{volume}{192}(\bibinfo{number}{1}), pp.
  \bibinfo{pages}{61--75}, \doi{10.1016/j.entcs.2007.08.017}.

\bibitemdeclare{inproceedings}{chaudhuri2015lightweight}
\bibitem{chaudhuri2015lightweight}
\bibinfo{author}{Kaustuv \surnamestart Chaudhuri\surnameend},
  \bibinfo{author}{Matteo \surnamestart Cimini\surnameend} \&
  \bibinfo{author}{Dale \surnamestart Miller\surnameend}
  (\bibinfo{year}{2015}): \emph{\bibinfo{title}{A lightweight formalization of
  the metatheory of bisimulation-up-to}}.
\newblock In: {\sl \bibinfo{booktitle}{Proceedings of the 2015 Conference on
  Certified Programs and Proofs}}, \bibinfo{organization}{ACM}, pp.
  \bibinfo{pages}{157--166}, \doi{10.1145/2676724.2693170}.

\bibitemdeclare{article}{church1940formulation}
\bibitem{church1940formulation}
\bibinfo{author}{Alonzo \surnamestart Church\surnameend}
  (\bibinfo{year}{1940}): \emph{\bibinfo{title}{A formulation of the simple
  theory of types}}.
\newblock {\sl \bibinfo{journal}{The journal of symbolic logic}}
  \bibinfo{volume}{5}(\bibinfo{number}{2}), pp. \bibinfo{pages}{56--68},
  \doi{10.2307/2266170}.

\bibitemdeclare{article}{cleaveland1993concurrency}
\bibitem{cleaveland1993concurrency}
\bibinfo{author}{Rance \surnamestart Cleaveland\surnameend},
  \bibinfo{author}{Joachim \surnamestart Parrow\surnameend} \&
  \bibinfo{author}{Bernhard \surnamestart Steffen\surnameend}
  (\bibinfo{year}{1993}): \emph{\bibinfo{title}{The Concurrency Workbench: A
  semantics-based tool for the verification of concurrent systems}}.
\newblock {\sl \bibinfo{journal}{ACM Transactions on Programming Languages and
  Systems (TOPLAS)}} \bibinfo{volume}{15}(\bibinfo{number}{1}), pp.
  \bibinfo{pages}{36--72}, \doi{10.1145/151646.151648}.

\bibitemdeclare{inproceedings}{compton2005embedding}
\bibitem{compton2005embedding}
\bibinfo{author}{Michael \surnamestart Compton\surnameend}
  (\bibinfo{year}{2005}): \emph{\bibinfo{title}{Embedding a fair CCS in
  Isabelle/HOL}}.
\newblock In: {\sl \bibinfo{booktitle}{Theorem Proving in Higher Order Logics:
  Emerging Trends Proceedings}}, p.~\bibinfo{pages}{30}, \doi{10.1.1.105.834}.
\newblock
  \urlprefix\url{https://web.comlab.ox.ac.uk/techreports/oucl/RR-05-02.pdf#page=36}.

\bibitemdeclare{inproceedings}{DurierHS17}
\bibitem{DurierHS17}
\bibinfo{author}{Adrien \surnamestart Durier\surnameend},
  \bibinfo{author}{Daniel \surnamestart Hirschkoff\surnameend} \&
  \bibinfo{author}{Davide \surnamestart Sangiorgi\surnameend}
  (\bibinfo{year}{2017}): \emph{\bibinfo{title}{{Divergence and Unique Solution
  of Equations}}}.
\newblock In \bibinfo{editor}{Roland \surnamestart Meyer\surnameend} \&
  \bibinfo{editor}{Uwe \surnamestart Nestmann\surnameend}, editors: {\sl
  \bibinfo{booktitle}{28th International Conference on Concurrency Theory
  (CONCUR 2017)}}, {\sl \bibinfo{series}{Leibniz International Proceedings in
  Informatics (LIPIcs)}}~\bibinfo{volume}{85}, \bibinfo{publisher}{Schloss
  Dagstuhl--Leibniz-Zentrum fuer Informatik}, \bibinfo{address}{Dagstuhl,
  Germany}, pp. \bibinfo{pages}{11:1--11:16},
  \doi{10.4230/LIPIcs.CONCUR.2017.11}.
\newblock \urlprefix\url{http://drops.dagstuhl.de/opus/volltexte/2017/7784}.

\bibitemdeclare{inproceedings}{van2005characterisation}
\bibitem{van2005characterisation}
\bibinfo{author}{Rob~J. \surnamestart van Glabbeek\surnameend}
  (\bibinfo{year}{2005}): \emph{\bibinfo{title}{A characterisation of weak
  bisimulation congruence}}.
\newblock In: {\sl \bibinfo{booktitle}{Processes, Terms and Cycles: Steps on
  the Road to Infinity}}, \bibinfo{publisher}{Springer}, pp.
  \bibinfo{pages}{26--39}, \doi{10.1007/11601548_4}.

\bibitemdeclare{book}{gordon1979edinburgh}
\bibitem{gordon1979edinburgh}
\bibinfo{author}{Michael J.~C. \surnamestart Gordon\surnameend},
  \bibinfo{author}{Arthur~J. \surnamestart Milner\surnameend} \&
  \bibinfo{author}{Christopher~P. \surnamestart Wadsworth\surnameend}
  (\bibinfo{year}{1979}): \emph{\bibinfo{title}{Edinburgh {LCF}: A Mechanised
  Logic of Computation}}.
\newblock {\sl \bibinfo{series}{Lecture Notes in Computer
  Science}}~\bibinfo{volume}{78}, \bibinfo{publisher}{Springer},
  \bibinfo{address}{Berlin Heidelberg}, \doi{10.1007/3-540-09724-4}.

\bibitemdeclare{article}{gorrieri2017ccs}
\bibitem{gorrieri2017ccs}
\bibinfo{author}{Roberto \surnamestart Gorrieri\surnameend}
  (\bibinfo{year}{2017}): \emph{\bibinfo{title}{CCS (25, 12) is
  Turing-complete}}.
\newblock {\sl \bibinfo{journal}{Fundamenta Informaticae}}
  \bibinfo{volume}{154}(\bibinfo{number}{1-4}), pp. \bibinfo{pages}{145--166},
  \doi{10.3233/FI-2017-1557}.

\bibitemdeclare{book}{Gorrieri:2015jt}
\bibitem{Gorrieri:2015jt}
\bibinfo{author}{Roberto \surnamestart Gorrieri\surnameend} \&
  \bibinfo{author}{Cristian \surnamestart Versari\surnameend}
  (\bibinfo{year}{2015}): \emph{\bibinfo{title}{{Introduction to Concurrency
  Theory}}}.
\newblock \bibinfo{series}{Transition Systems and CCS},
  \bibinfo{publisher}{Springer}, \bibinfo{address}{Cham},
  \doi{10.1007/978-3-319-21491-7}.

\bibitemdeclare{inproceedings}{hirschkoff1997full}
\bibitem{hirschkoff1997full}
\bibinfo{author}{Daniel \surnamestart Hirschkoff\surnameend}
  (\bibinfo{year}{1997}): \emph{\bibinfo{title}{A full formalisation of
  $\pi$-calculus theory in the calculus of constructions}}.
\newblock In: {\sl \bibinfo{booktitle}{International Conference on Theorem
  Proving in Higher Order Logics}}, \bibinfo{organization}{Springer}, pp.
  \bibinfo{pages}{153--169}, \doi{10.1007/BFb0028392}.

\bibitemdeclare{manual}{holdesc}
\bibitem{holdesc}
\bibinfo{author}{\surnamestart {HOL4 contributors}\surnameend}
  (\bibinfo{year}{2018}): \emph{\bibinfo{title}{{The HOL System DESCRIPTION}}}.
\newblock
  \urlprefix\url{http://sourceforge.net/projects/hol/files/hol/kananaskis-12/kananaskis-12-description.pdf}.

\bibitemdeclare{manual}{hollogic}
\bibitem{hollogic}
\bibinfo{author}{\surnamestart {HOL4 contributors}\surnameend}
  (\bibinfo{year}{2018}): \emph{\bibinfo{title}{{The HOL System LOGIC}}}.
\newblock
  \urlprefix\url{http://sourceforge.net/projects/hol/files/hol/kananaskis-12/kananaskis-12-logic.pdf}.

\bibitemdeclare{inproceedings}{kahsai2008implementing}
\bibitem{kahsai2008implementing}
\bibinfo{author}{Temesghen \surnamestart Kahsai\surnameend} \&
  \bibinfo{author}{Marino \surnamestart Miculan\surnameend}
  (\bibinfo{year}{2008}): \emph{\bibinfo{title}{Implementing spi calculus using
  nominal techniques}}.
\newblock In: {\sl \bibinfo{booktitle}{Conference on Computability in Europe}},
  \bibinfo{organization}{Springer}, pp. \bibinfo{pages}{294--305},
  \doi{10.1007/978-3-540-69407-6_33}.

\bibitemdeclare{inproceedings}{hurd2011opentheory}
\bibitem{hurd2011opentheory}
\bibinfo{author}{Joe \surnamestart Leslie-Hurd\surnameend}
  (\bibinfo{year}{2011}): \emph{\bibinfo{title}{The OpenTheory standard theory
  library}}.
\newblock In: {\sl \bibinfo{booktitle}{NASA Formal Methods Symposium}},
  \bibinfo{organization}{Springer}, pp. \bibinfo{pages}{177--191},
  \doi{10.1007/978-3-642-20398-5_14}.

\bibitemdeclare{article}{melham1992hol}
\bibitem{melham1992hol}
\bibinfo{author}{Thomas~F. \surnamestart Melham\surnameend}
  (\bibinfo{year}{1992}): \emph{\bibinfo{title}{The HOL pred_sets Library}}.
\newblock {\sl \bibinfo{journal}{Universiy of Cambridge Computer Lab}},
  \doi{10.1.1.219.5390}.
\newblock
  \urlprefix\url{http://citeseerx.ist.psu.edu/viewdoc/summary?doi=10.1.1.219.5390}.

\bibitemdeclare{article}{melham1994mechanized}
\bibitem{melham1994mechanized}
\bibinfo{author}{Thomas~F. \surnamestart Melham\surnameend}
  (\bibinfo{year}{1994}): \emph{\bibinfo{title}{A Mechanized Theory of the
  Pi-Calculus in HOL}}.
\newblock {\sl \bibinfo{journal}{Nord. J. Comput.}}
  \bibinfo{volume}{1}(\bibinfo{number}{1}), pp. \bibinfo{pages}{50--76},
  \doi{10.1.1.56.4370}.
\newblock \urlprefix\url{http://core.ac.uk/download/pdf/22878407.pdf}.

\bibitemdeclare{techreport}{milner1972logic}
\bibitem{milner1972logic}
\bibinfo{author}{Robin \surnamestart Milner\surnameend} (\bibinfo{year}{1972}):
  \emph{\bibinfo{title}{Logic for Computable Functions: description of a
  machine implementation}}.
\newblock \bibinfo{type}{Technical Report}, \bibinfo{institution}{STANFORD UNIV
  CA DEPT OF COMPUTER SCIENCE}.
\newblock \urlprefix\url{http://www.dtic.mil/dtic/tr/fulltext/u2/785072.pdf}.

\bibitemdeclare{book}{Mil89}
\bibitem{Mil89}
\bibinfo{author}{Robin \surnamestart Milner\surnameend} (\bibinfo{year}{1989}):
  \emph{\bibinfo{title}{{Communication and concurrency}}}.
\newblock \bibinfo{series}{{PHI} Series in computer science},
  \bibinfo{publisher}{Prentice-Hall}.

\bibitemdeclare{inproceedings}{mohamed1995mechanizing}
\bibitem{mohamed1995mechanizing}
\bibinfo{author}{Otmane~A{\"\i}t \surnamestart Mohamed\surnameend}
  (\bibinfo{year}{1995}): \emph{\bibinfo{title}{Mechanizing a $\pi$-calculus
  equivalence in HOL}}.
\newblock In: {\sl \bibinfo{booktitle}{International Conference on Theorem
  Proving in Higher Order Logics}}, \bibinfo{organization}{Springer}, pp.
  \bibinfo{pages}{1--16}, \doi{10.1007/3-540-60275-5_53}.

\bibitemdeclare{techreport}{Nesi:1992ve}
\bibitem{Nesi:1992ve}
\bibinfo{author}{Monica \surnamestart Nesi\surnameend} (\bibinfo{year}{1992}):
  \emph{\bibinfo{title}{{A formalization of the process algebra CCS in high
  order logic}}}.
\newblock \bibinfo{type}{Technical Report} \bibinfo{number}{UCAM-CL-TR-278},
  \bibinfo{institution}{University of Cambridge, Computer Laboratory}.
\newblock
  \urlprefix\url{http://www.cl.cam.ac.uk/techreports/UCAM-CL-TR-278.pdf}.

\bibitemdeclare{article}{Nesi:2017wo}
\bibitem{Nesi:2017wo}
\bibinfo{author}{Monica \surnamestart Nesi\surnameend} (\bibinfo{year}{1999}):
  \emph{\bibinfo{title}{{Formalising a Value-Passing Calculus in HOL}}}.
\newblock {\sl \bibinfo{journal}{Formal Aspects of Computing}}
  \bibinfo{volume}{11}(\bibinfo{number}{2}), pp. \bibinfo{pages}{160--199},
  \doi{10.1007/s001650050046}.

\bibitemdeclare{inproceedings}{norrish2013ordinals}
\bibitem{norrish2013ordinals}
\bibinfo{author}{Michael \surnamestart Norrish\surnameend} \&
  \bibinfo{author}{Brian \surnamestart Huffman\surnameend}
  (\bibinfo{year}{2013}): \emph{\bibinfo{title}{{Ordinals in HOL: Transfinite
  arithmetic up to (and beyond) $\omega_1$}}}.
\newblock In: {\sl \bibinfo{booktitle}{International Conference on Interactive
  Theorem Proving}}, \bibinfo{organization}{Springer}, pp.
  \bibinfo{pages}{133--146}, \doi{10.1007/978-3-642-39634-2_12}.

\bibitemdeclare{article}{parrow2009formalising}
\bibitem{parrow2009formalising}
\bibinfo{author}{Joachim \surnamestart Parrow\surnameend} \&
  \bibinfo{author}{Jesper \surnamestart Bengtson\surnameend}
  (\bibinfo{year}{2009}): \emph{\bibinfo{title}{Formalising the pi-calculus
  using nominal logic}}.
\newblock {\sl \bibinfo{journal}{Logical Methods in Computer Science}}
  \bibinfo{volume}{5}, \doi{10.2168/LMCS-5(2:16)2009}.

\bibitemdeclare{article}{pous2007new}
\bibitem{pous2007new}
\bibinfo{author}{Damien \surnamestart Pous\surnameend} (\bibinfo{year}{2007}):
  \emph{\bibinfo{title}{New up-to techniques for weak bisimulation}}.
\newblock {\sl \bibinfo{journal}{Theoretical Computer Science}}
  \bibinfo{volume}{380}, pp. \bibinfo{pages}{164--180},
  \doi{10.1016/j.tcs.2007.02.060}.

\bibitemdeclare{book}{theoryAndPractice}
\bibitem{theoryAndPractice}
\bibinfo{author}{Andrew~W. \surnamestart Roscoe\surnameend}
  (\bibinfo{year}{1998}): \emph{\bibinfo{title}{The theory and practice of
  concurrency}}.
\newblock \bibinfo{publisher}{{Prentice Hall}}.
\newblock
  \urlprefix\url{http://www.cs.ox.ac.uk/people/bill.roscoe/publications/68b.pdf}.

\bibitemdeclare{book}{RosUnder10}
\bibitem{RosUnder10}
\bibinfo{author}{Andrew~W. \surnamestart Roscoe\surnameend}
  (\bibinfo{year}{2010}): \emph{\bibinfo{title}{Understanding Concurrent
  Systems}}.
\newblock \bibinfo{publisher}{Springer}, \doi{10.1007/978-1-84882-258-0}.

\bibitemdeclare{book}{Sangiorgi:2011ut}
\bibitem{Sangiorgi:2011ut}
\bibinfo{author}{Davide \surnamestart Sangiorgi\surnameend}
  (\bibinfo{year}{2011}): \emph{\bibinfo{title}{Introduction to Bisimulation
  and Coinduction}}.
\newblock \bibinfo{publisher}{Cambridge University Press},
  \doi{10.1017/CBO9780511777110}.

\bibitemdeclare{inproceedings}{sangiorgi2015equations}
\bibitem{sangiorgi2015equations}
\bibinfo{author}{Davide \surnamestart Sangiorgi\surnameend}
  (\bibinfo{year}{2015}): \emph{\bibinfo{title}{Equations, contractions, and
  unique solutions}}.
\newblock In: {\sl \bibinfo{booktitle}{ACM SIGPLAN Notices}},
  \bibinfo{volume}{50}, \bibinfo{organization}{ACM}, pp.
  \bibinfo{pages}{421--432}, \doi{10.1145/2676726.2676965}.
\newblock \urlprefix\url{https://hal.inria.fr/hal-01089205}.

\bibitemdeclare{article}{sangiorgi2017equations}
\bibitem{sangiorgi2017equations}
\bibinfo{author}{Davide \surnamestart Sangiorgi\surnameend}
  (\bibinfo{year}{2017}): \emph{\bibinfo{title}{Equations, contractions, and
  unique solutions}}.
\newblock {\sl \bibinfo{journal}{ACM Transactions on Computational Logic
  (TOCL)}} \bibinfo{volume}{18}, p.~\bibinfo{pages}{4}, \doi{10.1145/2971339}.
\newblock \urlprefix\url{https://hal.inria.fr/hal-01647063}.

\bibitemdeclare{inproceedings}{sangiorgi1992problem}
\bibitem{sangiorgi1992problem}
\bibinfo{author}{Davide \surnamestart Sangiorgi\surnameend} \&
  \bibinfo{author}{Robin \surnamestart Milner\surnameend}
  (\bibinfo{year}{1992}): \emph{\bibinfo{title}{The problem of “Weak
  Bisimulation up to”}}.
\newblock In: {\sl \bibinfo{booktitle}{International Conference on Concurrency
  Theory}}, \bibinfo{organization}{Springer}, pp. \bibinfo{pages}{32--46},
  \doi{10.1007/BFb0084781}.
\newblock
  \urlprefix\url{http://citeseerx.ist.psu.edu/viewdoc/summary?doi=10.1.1.38.5964}.

\bibitemdeclare{book}{sangiorgi2011advanced}
\bibitem{sangiorgi2011advanced}
\bibinfo{author}{Davide \surnamestart Sangiorgi\surnameend} \&
  \bibinfo{author}{Jan \surnamestart Rutten\surnameend} (\bibinfo{year}{2011}):
  \emph{\bibinfo{title}{Advanced Topics in Bisimulation and Coinduction}}.
\newblock \bibinfo{publisher}{Cambridge University Press},
  \doi{10.1017/CBO9780511777110}.

\bibitemdeclare{inproceedings}{slind2008brief}
\bibitem{slind2008brief}
\bibinfo{author}{Konrad \surnamestart Slind\surnameend} \&
  \bibinfo{author}{Michael \surnamestart Norrish\surnameend}
  (\bibinfo{year}{2008}): \emph{\bibinfo{title}{A brief overview of HOL4}}.
\newblock In: {\sl \bibinfo{booktitle}{International Conference on Theorem
  Proving in Higher Order Logics}}, \bibinfo{organization}{Springer}, pp.
  \bibinfo{pages}{28--32}, \doi{10.1007/978-3-540-71067-7_6}.
\newblock
  \urlprefix\url{http://ts.data61.csiro.au/publications/nicta_full_text/1482.pdf}.

\bibitemdeclare{mastersthesis}{Tian:2017wrba}
\bibitem{Tian:2017wrba}
\bibinfo{author}{Chun \surnamestart Tian\surnameend} (\bibinfo{year}{2017}):
  \emph{\bibinfo{title}{{A Formalization of Unique Solutions of Equations in
  Process Algebra}}}.
\newblock Master's thesis, \bibinfo{school}{AlmaDigital},
  \bibinfo{address}{Bologna}.
\newblock \urlprefix\url{http://amslaurea.unibo.it/14798/}.

\end{thebibliography}
\end{document}